\documentclass[11pt,a4paper]{article}

\usepackage{latexsym}
\usepackage{amsfonts}
\usepackage{amssymb}
\usepackage{amsmath}
\usepackage{wasysym}
\usepackage[mathscr]{eucal}
\usepackage{theorem}
\usepackage{enumerate}
\usepackage{cite}
\usepackage{url}
\usepackage[dvips]{color}
\usepackage{graphicx}

\theoremstyle{plain}
\theorembodyfont{\itshape}
\newtheorem{thm}{Theorem}
\newtheorem{lem}{Lemma}
\newtheorem{prp}{Proposition}

\newtheorem{dfn}{Definition}

\theorembodyfont{\upshape}

\theorembodyfont{\upshape}

\newtheorem{rmk}{Remark}

\newcommand{\qed}{\hfill\mbox{\raggedright $\Box$}\medskip}

\newcommand{\mydate}{
 \ifcase\month \or
 January\or February\or March\or April\or May\or June\or
 July\or August\or September\or October\or November\or December\fi
 \space \number\year}
\newcommand{\vcirc}{\raisebox{1pt}{$\, \scriptstyle \circ \,$}}
\newcommand{\smin}{\,\raisebox{0.06em}{${\scriptstyle \in}$}\,}

\newcommand{\smwedge}{{\scriptstyle \wedge\,}}

\newcommand{\circledstar}{\mathbin{\mathchoice%
 {\ooalign{\hss$\displaystyle\ocircle$\hss\cr\hss$\displaystyle\star$\hss\cr}}%
 {\ooalign{\hss$\textstyle\ocircle$\hss\cr\hss$\textstyle\star$\hss\cr}}%
 {\ooalign{\hss$\scriptstyle\ocircle$\hss\cr\hss$\scriptstyle\star$\hss\cr}}%
 {\ooalign{\hss$\scriptscriptstyle\ocircle$\hss\cr%
           \hss$\scriptscriptstyle\star$\hss\cr}}}}
\newcommand{\bwedge}{\raisebox{0.2ex}{${\textstyle \bigwedge}$}}

\newcommand{\mathslf}[1]{\ensuremath{\mbox{\slshape\textsf{#1}}}}
\newcommand{\mathsmf}[1]{\ensuremath{\mbox{\footnotesize \slshape\textsf{#1}}}}

\begin{document}

\title{On Covariant Poisson Brackets \\
       in Classical Field Theory\,%
       \thanks{Work partially supported by CNPq (Conselho Nacional de
               Desenvolvimento Cient\'{\i}fico e Tecno\-l\'o\-gico), Brazil}}
\author{Michael Forger$^{\,1}$~~and~~M\'ario O.\ Salles$^{\,1,\,2}$
        \thanks{Work done in partial fulfillment of the requirements
                for the degree of Doctor in Science}}
\date{\normalsize
      $^1\,$ Instituto de Matem\'atica e Estat\'{\i}stica, \\
      Universidade de S\~ao Paulo, \\
      Caixa Postal 66281, \\
      BR--05315-970~ S\~ao Paulo, SP, Brazil  \\[4mm]
      $^2\,$ Centro de Ci\^encias Exatas e da Terra, \\
      Universidade Federal do Rio Grande do Norte, \\
      Campus Universit\'ario -- Lagoa Nova, \\
      BR--59078-970 Natal, RN, Brazil
      }
\maketitle

\thispagestyle{empty}

\begin{abstract}
\noindent
 How to give a natural geometric definition of a covariant Poisson bracket
 in classical field theory has for a long time been an open problem~-- as
 testified by the extensive literature on ``multi\-symplectic Poisson
 brackets'', together with the fact that all these proposals suffer
 from serious defects.
 On the other hand, the functional approach does provide a good
 candidate which has come to be known as the Peierls\,--\,De\,Witt
 bracket and whose construction in a geometrical setting is now
 well understood.
 Here, we show how the basic ``multisymplectic Poisson bracket''
 already proposed in the 1970s can be derived from the Peierls%
 \,--\,De\,Witt bracket, applied to a special class of functionals.
 This relation allows to trace back most (if not all) of the
 problems encountered in the past to ambiguities (the relation
 between differential forms on multiphase space and the functionals
 they define is not one-to-one) and also to the fact that this class
 of functionals does not form a Poisson subalgebra.
\end{abstract}

\begin{flushright}
 \parbox{12em}{
  \begin{center}
   Universidade de S\~ao Paulo \\
   RT-MAP-1203 \\
   Revised version \\
   January 2015
  \end{center}
 }
\end{flushright}

\newpage

\setcounter{page}{1}

\section{Introduction}

The quest for a fully covariant hamiltonian formulation of classical
field theory has a long history which can be traced back to the work
of Carath\'eodory~\cite{Ca}, De\,Donder~\cite{DD} and Weyl~\cite{We}
on the calculus of variations.
From a modern point of view, one of the main motivations is the issue of
quantization which, in traditional versions like canonical quantization
as well as more recent ones such as deformation quantization, starts
by bringing the classical theory into hamiltonian form. \linebreak
In the context of mechanics, where one is dealing with systems with
a finite number of degrees of freedom, this has led mathematicians to
develop entire new areas of differential geo\-metry, namely symplectic
geometry and then Poisson geometry, whereas physicists have been
motivated to embark on a more profound analysis of basic physical
concepts such as those of states and observables.
In the context of (relativistic) field theory, however, this is
not sufficient since, besides facing the formidable mathematical
problem of handling systems with an infinite number of degrees
of freedom, we have to cope with new physical principles, most
notably those of covariance and of locality.
The principle of covariance states that meaningful laws of physics
do not depend on the choice of (local) coordinates in space-time
employed in their formulation: extending the axiom of Lorentz
invariance in special relativity, it is one of the cornerstones of
general relativity and underlies the modern geometrical approach
to field theory as a whole.
Equally important is the principle of locality, stating that events
(including measurements of physical quantities) localized in regions
of space-time that are spacelike separated cannot exert any influence
on each other.
Clearly, a mathematically and physically correct hamiltonian formalism
for classical field theory should respect these principles: it should
be manifestly covariant and local, as is the modern algebraic approach
to quantum field theory; see, e.g., Ref.~\cite{BFV}.

As an example of a method that does not meet these requirements, we may
quote the standard hamiltonian formulation of classical field theory,
based on a functional formalism in terms of Cauchy data: there, the
mere necessity of choosing some Cauchy surface spoils covariance
from the very beginning!
To avoid that, a different approach is needed.

Over the last few decades, attempts to construct such a different approach
have produced a variety of proposals that, roughly speaking, can be assembled
into two groups.

One of these extends the geometrical tools that were so sucessful in
mechanics to the situation encountered in field theory by treating
spatial derivatives of fields on the same footing as time derivatives:
in the context of a first order formalism, as in mechanics, this
requires associating to each field component, say $\varphi^i$, not
just one canonically conjugate momentum $\, \pi_i^{} = \partial L /
\partial \dot{\varphi}^i \,$ but rather $n$ canonically conjugate
momenta $\, \pi_i^\mu = \partial L / \partial \, \partial_\mu \varphi^i$,
where $n$ is the dimension of space-time.
(In mechanics, time is the only independent variable, so $n\!=\!1$.)
Identifying the appropriate geometrical context has led to the
introduction of new geometrical entities now commonly referred
to as ``multisymplectic'' and/or ``polysymplectic'' structures,
and although their correct mathematical definition has only
recently been completely elucidated~\cite{FG}, the entire
circle of ideas surrounding them is already reasonably well
established, forming a new area of differential geometry;
see~\cite{Gaw,Kij,GoS,GuS,KS1,KS2,Gun,Ma1,Ma2,Got,GIM,CCI}
for early references.

A different line of thought is centered around the concept of
``covariant phase space''~\cite{CW,Cr,Zu}, defined as the space
of solutions of the equations of motion: using this space to
substitute the corresponding space of Cauchy data eliminates
the need to refer to a specific choice of Cauchy surface and
has the additional benefit of providing an embedding into the
larger space of all field configurations, allowing us to classify
statements as valid ``off shell'' (i.e., on the entire space of
field configurations) or ``on shell'' (i.e., only on the subspace
of solutions of the equations of motion).

Each of the two methods, the \emph{multi\-symplectic formalism}
as well as the \emph{covariant functional formalism}, has its
own merits and its own drawbacks, and experience has shown that
best results are obtained by appropriately combining them.

As an example to demonstrate how useful the combination of these two
approaches can become, we shall in the present paper discuss the
problem of giving an appropriate definition of the Poisson bracket,
or better, the \emph{covariant Poisson bracket}.
From the point of view of quantization, this is a question of
fundamental importance, given the fact that the Poisson bracket
is expected to be the classical limit of the commutator in quantum
field theory.
Moreover, quantum field theory provides compelling motivation for
discussing this limit in a covariant setting, taking into account
that the (non-covariant) equal-time Poisson brackets of the standard
hamiltonian formulation of classical field theory would correspond,
in the sense of a classical limit, to the (non-covariant) equal-%
time commutators of quantum field theory, which are known not to
exist in interacting quantum field theories, due to Haag's theorem.

Unfortunately, in the context of the multisymplectic formalism, the
status of covariant \linebreak Poisson brackets is highly unsatisfactory.
This may come as a bit of a surprise, given the beautiful and conceptually
simple situation prevailing in mechanics, where the existence of a Poisson
bracket on the algebra $C^\infty(P)$ of smooth functions on a manifold~$P$
is equivalent to the statement that $P$ is a Poisson manifold and, as such,
qualifies as a candidate for the phase space of a classical hamiltonian
system: for any such system, the algebra of observables is just the
Poisson algebra $C^\infty(P)$ itself or, possibly, an appropriate
subalgebra thereof, and the space of pure states is just the
Poisson manifold~$P$ itself.
In particular, this is true in the special case when $P$ is a \emph%
{symplectic manifold}, with symplectic form~$\omega$, say, and where
the Poisson bracket of two functions $\, f,g \in C^\infty(P) \,$ is
the function  $\, \{f,g\} \in C^\infty(P) \,$ defined by
\begin{equation} \label{eq:POISB1}
 \{f,g\}~=~i_{X_g}^{} i_{X_f}^{} \, \omega~=~\omega(X_f,X_g)~,
\end{equation}
where $\, X_f \in \mathfrak{X}(P) \,$ denotes the hamiltonian vector field
associated with $\, f \in C^\infty(P)$, i.e.,
\begin{equation} \label{eq:HAMVF1}
 i_{X_f}^{} \, \omega~=~df~.
\end{equation}

This situation changes considerably, and for the worse, when we pass to
the multisymplectic setting, where $\omega$ is no longer a $2$-form but
rather an $(n\!+\!1)$-form and the hamiltonian vector field $X_f$ is no
longer associated with a function $f$ but rather with an $(n\!-\!1)$-%
form~$f$, $n$ being the dimension of space-time.
It can then be shown that equation~(\ref{eq:HAMVF1}) imposes restrictions
not only on the type of vector field that is allowed on its lhs but also
on the type of differential form that is allowed on its rhs.
Indeed, the validity of an equation of the form $\, i_X^{} \omega = df \,$
implies that the vector field $X$ must be locally hamiltonian, i.e., we
have $\, L_X \omega = 0$, but it also implies that the form~$f$ must be
hamiltonian, which by definition means that its exterior derivative~$df$
must vanish on all multivectors of degree~$n$ whose contraction with the
$(n+1)$-form $\omega$ is zero, and this is a non-trivial condition as
soon as $n\!>\!1$.
(It is trivial for $n\!=\!1$ since $\omega$ is assumed to be non-degenerate.)
Thus it is only on the space $\Omega_H^{n-1}(P)$ of hamiltonian $(n-1)$-%
forms that equation~(\ref{eq:POISB1}) provides a reasonable candidate
for a Poisson bracket.
(Note, however, that as a Lie algebra with respect to such a bracket,
$\Omega_H^{n-1}(P)$ would have a huge center, containing the entire space
$Z^{n-1}(P)$ of closed $(n\!-\!1)$-forms on~$P$, since the linear map from
$\Omega_H^{n-1}(P)$ to $\mathfrak{X}_{LH}(P)$ that takes $f$ to $X_f$ is far
from being one-to-one: its kernel is precisely $Z^{n-1}(P)$.)
Anyway, the argument suggests that the transition from mechanics to field
theory should somehow involve a replacement of functions by differential
forms of degree $n\!-\!1$~-- which is not completely unreasonable when we
consider the fact that, in field theory, conservation laws are formulated
in terms of conserved currents, which are closed $(n\!-\!1)$-forms.

Unfortunately, this replacement leads to a whole bunch of serious problems,
some of which are insurmountable.
First and foremost, there is no reasonable candidate for an associative
product on the space $\Omega_H^{n-1}(P)$ which would provide even a
starting point for defining a Poisson algebra.
Second, as has been observed repeatedly in the literature~%
\cite{Gaw,Kij,Ka1,Ka2,HK1,HK2,Hel}, the condition of being a
locally hamiltonian vector field or a hamiltonian $(n\!-\!1)$-%
form forces these objects to depend at most linearly on the
multimomentum variables, and moreover we can easily think of
observables that are associated to forms of other degree
(such as a scalar field, corresponding to a $0$-form, or
the electromagnetic field strength tensor, corresponding
to a $2$-form): this by itself provides enough evidence
to conclude that hamiltonian $(n\!-\!1)$-forms constitute
an extremely restricted class of observables and that setting
up an adequate framework for general observables will require
going beyond this domain.
And finally, as has already also been noted long ago \cite{Gaw,GoS,GuS,%
Kij,KS1,KS2}, the multisymplectic Poisson bracket defined by equation~%
(\ref{eq:POISB1}) fails to satisfy the Jacobi identity.
In the case of an exact multisymplectic form ($\omega = - d\theta$),
this last problem can be cured by modifying the defining equation~%
(\ref{eq:POISB1}) through the addition of an exact (hence closed)
term, as follows~\cite{FHR}:
\begin{equation} \label{eq:POISB2}
 \{f,g\}~=~i_{X_g}^{} i_{X_f}^{} \omega \, + \,
           d \, \Bigl( i_{X_g}^{} f \, - \, i_{X_f}^{} g \, - \,
                       i_{X_g}^{} i_{X_f}^{} \theta \Bigr)~.
\end{equation}

\noindent
However, this does not settle any of the other two issues, namely
\vspace{-1mm}
\begin{itemize}
 \item the lack of an associative product to construct a Poisson algebra;%
       \vspace{-1mm}
 \item the restriction to hamiltonian forms and forms of degree $n\!-\!1$,
       which leads to unreasonable constraints on the observables that are
       allowed, excluding some that appear naturally in physicists'
       calculations.
       \vspace{-1mm}
\end{itemize}
It should be mentioned here that these are long-standing problems: they
have been recognized since the early stages of development of the subject
(see, e.g., \cite{Gaw,Kij} and also~\cite{Ka1,Ka2}) but have so far
remained unsolved.

A simple idea in this direction that has already been exploited is based
on the observation that differential forms do admit a natural associative
product, namely the wedge product, so one may ask what happens if, in the
above construction, vector fields are replaced by multivector fields and
$(n\!-\!1)$-forms by forms of arbitrary degree.
As it turns out, this leads to a modified super-Poisson bracket, defined
by a formula analogous to equation~(\ref{eq:POISB2})~\cite{FPR1,FPR2}.
But it does not help to overcome either of the aforementioned other
two issues.

On the other hand, in the context of the covariant functional formalism,
there is an obvious associative and commutative product, namely just the
pointwise product of functionals, and apart from that, there also exists
a natural and completely general definition of a covariant Poisson bracket
such that, when both are taken together, all the properties required of a
Poisson algebra are satisfied: this bracket is known as the Peierls\,--\,%
De\,Witt bracket~\cite{Pe,DW1,DW2,DW3,FSR}.

Thus the question arises as to what might be the relation, if any, between
the covariant functional Poisson bracket, or Peierls\,--\,De\,Witt bracket,
and the various candidates for multi\-symplectic Poisson brackets that
have been discussed in the literature, among them the ones written down
in equations~(\ref{eq:POISB1}) and~(\ref{eq:POISB2}).
That is the question we shall address in this paper.

In the remainder of this introduction, we want to briefly sketch
the answer proposed here: details will be filled in later on.
Starting out from the paradigm that, mathematically, classical fields are
to be described by sections of fiber bundles, suppose we are given a fiber
bundle~$P$ over a base manifold~$M$, where $M$ represents space-time,
with projection $\, \rho: P \longrightarrow M$, and suppose that the
classical fields appearing in the field theoretical model under study
are sections $\, \phi: M \longrightarrow P \,$ of~$P$ (i.e., maps
$\phi$ from~$M$ to~$P$ satisfying $\, \rho \circ \phi = \mathrm{id}_M$),
subject to appropriate regularity conditions: for the sake of definite%
ness, we shall assume here that all manifolds and bundles are ``regular''
in the sense of being smooth, while the regularity of sections may vary
between smooth ($C^\infty$) and distributional ($C^{-\infty}$).%
\footnote{If no specification is given, it is tacitly assumed that we
are dealing with smooth sections.}
To fix terminology, we define, for any section $f$ of any vector
bundle~$V$ over~$P$, its \emph{base support}\/ or \emph{space-time
support}\/, denoted here by $\mathrm{supp} f$, to be the closure
of the set of points in~$M$ such that the restriction of~$f$ to
the corresponding fibers of~$P$ does not vanish identically, i.e.,%
\footnote{Using the abbreviation ``supp'' for the base support rather
than the ordinary support (which would be a subset of~$P$) constitutes
a certain abuse of language, but will do no harm since the ordinary
support will play no role in this paper.}
\begin{equation} \label{eq:BASSUPP}
 \mathrm{supp} f~=~\overline{\bigl\{ x \in M~|~f\vert_{P_x} \neq 0
   \bigr\}}~.
\end{equation}
Now suppose that $f$ is a differential form on~$P$ of degree~$p$,
say, and that $\Sigma$ is a closed $p$-dimensional submanifold
of~$M$, possibly with boundary, subject to the restriction that
$\Sigma$ and $\mathrm{supp} f$ should have compact intersection%
\footnote{Of course, this restriction is automatically satisfied
if $\Sigma$ is compact and also if $f$ has compact base support.
Moreover, given an arbitrary differential form~$f$ on~$P$, we can
always construct one with compact base support by multiplying with
a ``cutoff function'', i.e., the pull-back to~$P$ of a function
of compact support on~$M$.}
so as to guarantee that the following integral is well defined,
providing a functional $\mathcal{F}_{\Sigma,f}^{}$ on the space of
sections of~$P$,
\begin{equation} \label{eq:LFUNCT1}
 \mathcal{F}_{\Sigma,f}^{}[\phi]~=~\int_\Sigma \phi^* f~,
\end{equation}
where $\phi^* f$ is of course the pull-back of~$f$ to~$M$ via~$\phi$.
Regarding boundary conditions, we shall usually require that if $\Sigma$
has a boundary, it should not intersect the base support of~$f$:
\begin{equation} \label{eq:BOUNDC}
 \partial\Sigma \cap \mathrm{supp} f~=~\emptyset~.
\end{equation}

\noindent
This simple construction provides an especially interesting class of
functionals for various reasons, the most important of them being the
fact that they are \emph{local}\/, since they are simple integrals,
over regions or more general submanifolds of space-time, of local
densities such as, e.g., polynomials of the basic fields and their
derivatives, up to a certain order.%
\footnote{To incorporate derivatives up to order~$r$, say, of fields
that are sections of some fiber bundle~$E$ over~$M$, one has to define
$P$ using the $r$-th order jet bundle $J^r E$ of~$E$.}
This is an intuitive notion of locality for functionals of classical
fields, but as has been shown recently, it can also be formulated in
mathematically rigorous terms~\cite{BFR}.
Either way, it is clear that the product of two local functionals of
the form~(\ref{eq:LFUNCT1}) is no longer a local functional of the
same form: rather, we get a ``bilocal'' functional associated with a
submanifold of~$M \times M$ and a differential form on~$P \times P$.
Therefore, a mathematically interesting object to study would be the
algebra of ``multilocal'' functionals which is generated by the local
ones, much in the same way as, on an ordinary vector space, the algebra
of polynomials is generated by the monomials.

But the point of main interest for our work appears when we assume $P$
to be a multi\-symplectic fiber bundle~\cite{FG} and $M$ to be a Lorentz
manifold, usually satisfying some additional hypotheses regarding its
causal structure: more specifically, we shall assume $M$ to be globally
hyperbolic since this is the property that allows us to speak of Cauchy
surfaces.
In fact, as is now well known, $M$ will in this case admit a foliation
by Cauchy surfaces defined as the level sets of some smooth time function.
However, it is often convenient not to fix any metric on~$M$ ``a priori''
since, in the context of general relativity, the space-time metric
itself is a dynamical entity and not a fixed background field.
Within this context, and for the case of a regular first-order
hamiltonian system where fields are sections of a given configuration
bundle~$E$ over~$M$ and the dynamics is obtained from a regular first-%
order lagrangian via Legendre transform, it has been shown in~\cite{FSR}
that one may use that structure to define the Peierls\,--\,De\,Witt bracket
as a functional Poisson bracket on covariant phase space.
Here, we want to show how, in the same context, \emph{multisymplectic
Poisson brackets between forms}, such as in equations~(\ref{eq:POISB1})
and~(\ref{eq:POISB2}), \emph{can be derived from the Peierls\,--\,De\,Witt
bracket between the corresponding functionals}\/.
For the sake of simplicity, this will be done for the case of
$(n\!-\!1)$-forms, but we expect similar arguments to work in any degree.

Concretely, we shall prove that given a fixed hypersurface $\Sigma$ in~$M$
(typically, a Cauchy surface) and two hamiltonian $(n\!-\!1)$-forms $f$
and $g$, we have
\begin{equation} \label{eq:PDWBMSPB}
 \bigl\{ \mathcal{F}_{\Sigma,f}^{} \,,\, \mathcal{F}_{\Sigma,g}^{} \bigr\}~
 =~\mathcal{F}_{\Sigma,\{f,g\}}^{}~,
\end{equation}
where the bracket on the lhs is the Peierls\,--\,De\,Witt bracket of
functionals and the bracket $\{f,g\}$ that appears on the rhs is a
``multisymplectic pseudo-bracket'' or ``multisymplectic bracket''
given by a formula analogous to equation~(\ref{eq:POISB1}) or to
equation~(\ref{eq:POISB2}).
A more detailed explanation of this result will be deferred to the main
body of the paper~-- last but not least because the construction requires
the systematic use of both types of multiphase space that appear in
field theory and that we refer to as ``ordinary multiphase space''
and ``extended multiphase space'', respectively: they differ in that
the latter is a one-dimensional extension of the former, obtained by
including an additional scalar ``energy type'' variable.
Geometrically, extended multiphase space is an affine line bundle
over ordinary multiphase space, and the hamiltonian $\mathcal{H}$
of any theory with this type of ``field content'' is a section of
this affine line bundle.
Moreover, each of these two multiphase spaces comes equipped with a
multisymplectic structure which is exact (i.e., the multisymplectic
form is, up to a sign introduced merely for convenience, the exterior
derivative of a multicanonical form), naturally defined as follows.
First, one constructs the multisymplectic form~$\omega$ and the
multicanonical form~$\theta$ on the extended multiphase space by
means of a procedure that can be thought of as a generalization
of the construction of the symplectic structure on the cotangent
bundle of an arbitrary manifold.
Then, the corresponding forms on the ordinary multiphase space are
obtained from the previous ones by pull-back via the hamiltonian~%
$\mathcal{H}$: therefore, they will in what follows be denoted by~%
$\omega_{\mathcal{H}}^{}$ and by~$\theta_{\mathcal{H}}^{}$ to indicate
their dependence on the choice of hamiltonian.
We can express this by saying that the multisymplectic structure on
extended multiphase space is ``kinematical'' whereas that on ordinary
multiphase space is ``dynamical''.
Correspondingly, we shall refer to the brackets defined by equations~%
(\ref{eq:POISB1}) and~(\ref{eq:POISB2}) on extended multiphase space as
``kinematical multisymplectic brackets'' and to the brackets defined by
the analogous equations
\begin{equation} \label{eq:POISB3}
 \{f,g\}~=~i_{X_g}^{} i_{X_f}^{} \, \omega_{\mathcal{H}}^{}~
         =~\omega_{\mathcal{H}}^{}(X_f,X_g)
\end{equation}
with
\begin{equation} \label{eq:HAMVF2}
 i_{X_f}^{} \, \omega_{\mathcal{H}}^{}~=~df
\end{equation}
and
\begin{equation} \label{eq:POISB4}
 \{f,g\}~=~i_{X_g}^{} i_{X_f}^{} \omega_{\mathcal{H}}^{} \, + \,
           d \, \Bigl( i_{X_g}^{} f \, - \, i_{X_f}^{} g \, - \,
                       i_{X_g}^{} i_{X_f}^{} \theta_{\mathcal{H}}^{} \Bigr)
\end{equation}
as ``dynamical multisymplectic brackets''.
In both cases, the brackets defined by the simpler formulas~(\ref{eq:POISB1})
and~(\ref{eq:POISB3}) are really only ``pseudo-brackets'' because they fail
to satisfy the Jacobi identity, and the correction terms that appear in
equations~(\ref{eq:POISB2}) and~(\ref{eq:POISB4}) are introduced to cure
this defect.
At any rate, what appears on the rhs of equation~(\ref{eq:PDWBMSPB}) above
is the dynamical bracket on ordinary multiphase space and not the kinematical
bracket on extended multiphase space~-- in accordance with the fact that the
Peierls\,--\,De\,Witt bracket itself is dynamical.

We conclude this introduction with a few comments about the organization
of the paper.
In~Section~2, we set up the geometric context for the functional calculus
in classical field theory, introduce the class of local functionals to be
investigated and give an explicit formula for their first functional
derivative.
In~Section~3, we present a few elementary concepts from multisymplectic
geometry, which is the adequate mathematical background for the covariant
hamiltonian formulation of classical field theory.
In~Section~4, we combine the two previous sections to formulate, in this
context, the variational principle that provides the dynamics and derive
not only the equations of motion (De\,Donder\,--\,Weyl equations) but also
their linearization around a given solution (linearized De\,Donder\,--\,Weyl
equations), with emphasis on a correct treatment of boundary conditions.
In~Section~5, we present the classification of locally hamiltonian and
exact hamiltonian vector fields on multiphase space.
Section~6 contains the main result of the paper on the connection between
multisymplectic Poisson brackets and the functional Poisson bracket of
Peierls and De\,Witt.
Finally, Section~7 provides further discussion of this result, its
implications and perspectives for future investigations.

The paper presents a substantially revised and expanded version of the
main results contained in the PhD thesis of the second author~\cite{Sa},
which was elaborated under the supervision of the first author.

\section{Geometric setup for the functional calculus}

We begin by collecting some concepts and notations that we use throughout
the article. \linebreak
As already mentioned in the introduction, classical fields are sections
of fiber bundles over space-time, so our starting point will be to fix
a fiber bundle $P$ over the space-time manifold~$M$ (not necessarily
endowed with a fixed metric, as mentioned before), with projection
$\, \rho: P \longrightarrow M$.
The space of field configurations $\mathscr{C}$ is then the space of
(smooth) sections of~$P$, or an appropriate subspace thereof,
\begin{equation} \label{eq:FCONFS1} 
 \mathscr{C}~\subset~\Gamma^\infty(P)~,
\end{equation}
whose elements will, typically, be denoted by $\phi$.
Formally, we can view this space as a mani\-fold which, at each point $\phi$,
has a tangent space $T_\phi^{} \mathscr{C}$ and, similarly, a cotangent space
$T_\phi^* \mathscr{C}$. \linebreak
Explicitly, denoting by $V_\phi^{}$ the vertical bundle of~$P$, pulled back
to~$M$ via~$\phi$,
\begin{equation}
 V_\phi^{}~=~\phi^*(\mathrm{Ver} P)~,
\end{equation}
and by $V_\phi^{\circledast}$ its twisted dual, defined by taking the tensor
product of its ordinary dual with the line bundle of volume forms over
the base space,
\begin{equation}
 V_\phi^{\circledast}~=~V_\phi^* \otimes \bwedge^n T^* M~,
\end{equation}
we have that, according to the principles of the variational calculus,
$T_\phi^{} \mathscr{C}$ is the space of smooth sections of $V_\phi^{}$,
or an appropriate subspace thereof,
\begin{equation} \label{eq:FTANS1} 
 T_\phi^{} \mathscr{C}~\subset~\Gamma^\infty(V_\phi^{})~,
\end{equation}
whose elements will, typically, be denoted by $\delta\phi$ and called
variations of~$\phi$, whereas  $T_\phi^* \mathscr{C}$ is the space of
distributional sections of~$V_\phi^{\circledast}$, or an appropriate sub%
space thereof,
\begin{equation} \label{eq:FCOTS1} 
 T_\phi^* \mathscr{C}~\subset~\Gamma^{-\infty}(V_\phi^{\circledast})~.
\end{equation}
The reader will note that in equations~(\ref{eq:FCONFS1}), (\ref{eq:FTANS1})
and (\ref{eq:FCOTS1}), we have required only inclusion, rather than equality.
One reason is that the system may be subject to constraints on the fields
which cannot be reduced to the simple statement that they should take
values in some appropriate subbundle of the original bundle (this case
could be handled by simply changing the choice of the bundle~$P$).
But even for unconstrained systems, which are the only ones that we
shall be dealing with in this paper, there is another reason, namely
that we have not yet specified the support properties to be employed.
One obvious possibility is to set
\begin{equation} \label{eq:FTCTS1} 
 T_\phi^{} \mathscr{C}~=~\Gamma_c^\infty(V_\phi^{})~~,~~
 T_\phi^* \mathscr{C}~=~\Gamma^{-\infty}(V_\phi^{\circledast})~,
\end{equation}
which amounts to allowing only variations with compact support.
At the other extreme, we may set
\begin{equation} \label{eq:FTCTS2} 
 T_\phi^{} \mathscr{C}~=~\Gamma^\infty(V_\phi^{})~~,~~
 T_\phi^* \mathscr{C}~=~\Gamma_c^{-\infty}(V_\phi^{\circledast})~.
\end{equation}
And finally, there is a third option, specifically adapted to the
situation where the base space is a globally hyperbolic lorentzian
manifold and adopted in~\cite{FSR}, which is to take
\begin{equation} \label{eq:FTCTS3} 
 T_\phi^{} \mathscr{C}~=~\Gamma_{\mathrm{sc}}^\infty(V_\phi^{})~~,~~
 T_\phi^* \mathscr{C}~=~\Gamma_{\mathrm{tc}}^{-\infty}(V_\phi^{\circledast})~,
\end{equation}
where the symbols ``sc'' and ``tc'' indicate that the sections
are required to have \emph{spatially compact support}\/ and
\emph{temporally compact support}\/, respectively.
These options correspond to different choices for the support
properties of the functionals that will be allowed.

Generally speaking, given a functional $\mathslf{F}\,$ on $\mathscr{C}$,
we define its \emph{base support}\/ or \emph{space-time support}\/,
denoted here simply by $\, \mathrm{supp}\, \mathslf{F} \,$, as follows
\cite{BFR}:
\begin{equation} \label{eq:FUNSUP1}
 x \notin \mathrm{supp}\, \mathslf{F}~~\Longleftrightarrow~~
  \begin{array}{c}
   \mbox{There exists an open neighborhood $U_x$ of~$x$ in~$M$} \\
   \mbox{such that for any two field configurations
         $\, \phi_1^{},\phi_2^{} \in \mathscr{C}$} \\
   \mbox{satisfying $\, \phi_1^{} = \phi_2^{} \,$ on $\, M \setminus U_x \,$,
         $\mathslf{F}\,[\phi_1^{}] = \mathslf{F}\,[\phi_2^{}]$.}
  \end{array}
\end{equation}
This definition implies that $\, \mathrm{supp}\, \mathslf{F} \,$ is
a closed subset of~$M$ since its complement is open: it is the largest
open subset of~$M$ such that, intuitively speaking, $\mathslf{F}\,$ is
\emph{insensitive to variations of its argument localized within that
open subset}.
It also implies that the functional derivative of~$\mathslf{F}$ \linebreak
(if it exists) satisfies
\begin{equation} \label{eq:FUNSUP2}
 \mathslf{F}^{\>\prime}[\phi] \cdot \delta\phi~=~0
 \qquad \mbox{if $\; \mathrm{supp} \, \mathslf{F} \,\cap\,
                     \mathrm{supp} \, \delta\phi~=~\emptyset$}~.
\end{equation}
For later use, we note that the functional derivative will often be
expressed in terms of a (formal) variational derivative:
\begin{equation} \label{eq:FUNDER}
 \mathslf{F}^{\>\prime}[\phi] \cdot \delta\phi~
 =~\int_M d^{\,n} x~\frac{\delta\mathslf{F}}{\delta\phi^i}[\phi](x) \;
   \delta\phi^i(x)~.
\end{equation}
Typically, as always in distribution theory, the functional derivative
will be well defined on variations $\delta\phi$ such that $\; \mathrm{supp}
\, \mathslf{F} \,\cap\, \mathrm{supp} \, \delta\phi \;$ is compact.
Thus if no restrictions on the space-time support of~$\mathslf{F}\,$
are imposed, we must adopt the choice made in equation~(\ref{eq:FTCTS1}).
At the other extreme, if the space-time support of~$\mathslf{F}\,$ is
supposed to be compact, we may adopt the choice made in equation~%
(\ref{eq:FTCTS2}).
And finally, the choice made in equation~(\ref{eq:FTCTS3}) is the
adequate one for dealing with functionals that have \emph{temporally
compact support}\/, i.e., space-time support contained in the inverse
image of a bounded interval in~$\mathbb{R}$ under some global time
function: the typical example is that of a local functional of the
form given by equation~(\ref{eq:LFUNCT1}) when $\Sigma$ is some
Cauchy surface.
More generally, note that for local functionals of
the form given by equation~(\ref{eq:LFUNCT1}), we have
\begin{equation} \label{eq:LFUNCT2}
 \mathrm{supp} \, \mathcal{F}_{\Sigma,f}^{}~
 =~\Sigma \,\cap\, \mathrm{supp} f~.
\end{equation}
However, it should not be left unnoticed that the equality in
equations~(\ref{eq:FTCTS1})-(\ref{eq:FTCTS3}) and, possibly, in
equation~(\ref{eq:LFUNCT2}), can only be guaranteed to hold for
non-degenerate systems, since in the case of degenerate systems,
there will be additional constraints implying that we must return
to the option of replacing equalities by inclusions, as before.

In what follows, we shall make extensive use of the fact that variations
of sections can always be witten as compositions with projectable vector
fields, or even with vertical vector fields, on the total space~$P$.
To explain this, recall that a vector field $X$ on the total space
of a fiber bundle is called \emph{projectable} if the tangent map
to the bundle projection takes the values of~$X$ at any two points
in the same fiber to the same tangent vector at the corresponding
base point, i.e.,
\begin{equation} \label{eq:PROJVF1}
 T_{p_1^{}}^{} \rho \cdot X(p_1^{})~=~T_{p_2^{}}^{} \rho \cdot X(p_2^{})
 \qquad \mbox{for $\, p_1^{},p_2^{} \in P \,$
              such that $\, \rho(p_1^{}) = \rho(p_2^{})$}~.
\end{equation}
This is equivalent to requiring that there exists a vector field $X_M^{}$
on the base space which is $\rho\,$-related to~$X$,
\begin{equation} \label{eq:PROJVF2}
 X_M^{}(m)~=~T_{p\,}^{} \rho \cdot X(p)
 \qquad \mbox{for $\, p \in P \,$ such that $\, \rho(p) = m$}~.
\end{equation}
In particular, $X$ is called \emph{vertical} if $X_M$ vanishes.
Now note that given any projectable vector field $X$ on~$P$, we
obtain a functional vector field $\mathslf{X}$ on $\mathscr{C} \,$
whose value at each point $\, \phi \in \mathscr{C} \,$ is the
functional tangent vector $\, \mathslf{X}[\phi] \in T_\phi
\mathscr{C}$, denoted in what follows by $\delta_X^{} \phi$,
defined as
\begin{equation} \label{eq:VFREPVAR1} 
 \delta_X^{} \phi~=~X(\phi) \, - \, T\phi \, (X_M^{})~,
\end{equation}
or more explicitly,
\begin{equation} \label{eq:VFREPVAR2} 
 \delta_X^{} \phi(m)~=~X(\phi(m)) \, - \, T_m^{} \phi \, (X_M^{}(m))
 \qquad \mbox{for $\, m \in M$}~.
\end{equation}
Conversely, it can be shown that every functional tangent vector can
be obtained in this way from a vertical vector field $X$ on~$P$, i.e.,
given a section $\delta\phi$ of $\phi^*(\mathrm{Ver} P)$, there exists
a vertical vector field $X$ on~$P$ representing it in the sense that
$\delta\phi$ is equal to $\delta_X^{} \phi$.
(To do so, we can apply the implicit function theorem to construct, for
any point $m$ in~$M$, a system of local coordinates $(x^\mu,y^\alpha)$
for~$P$ around~$\phi(m)$ in which $\rho$ corresponds to the projection
onto the first factor, $(x^\mu,y^\alpha) \mapsto x^\mu$, and $\phi$
corresponds to the embedding in the first factor, $x^\mu \mapsto
(x^\mu,0)$.
Moreover, in these coordinates, $\delta\phi$ is given by functions
$\delta\phi^\alpha(x^\mu)$ whereas $X$ is given by functions $X^\alpha%
(x^\mu,y^\beta)$, so we may simply define an extension of the former
to the latter by requiring the $X^\alpha$ to be independent of the
$y^\beta$, setting $\, X^\alpha(x^\mu,y^\beta) = \delta\phi^\alpha(x^\mu) \,$
in a neighborhood of the origin in $y$-space and then using a smooth
cutoff function in $y$-space.)

Of course, the reader may wonder why, in this context, we bother to
allow for projectable vector fields rather than just vertical ones.
The point is that although vertical vector fields are entirely
sufficient to represent variations of sections, we shall often
encounter the converse situation in which we consider variations
of sections induced by vector fields which are \emph{not}\/ vertical
but only projectable, such as the hamiltonian vector fields appearing
in equations \linebreak (\ref{eq:POISB1})-(\ref{eq:POISB2}) and
(\ref{eq:POISB3})-(\ref{eq:POISB4}).

Regarding notation, we shall often think of a projectable vector field
as a pair $\, X = (X_P,X_M)$ \linebreak consisting of a vector field $X_P$
on the total space~$P$ and a vector field $X_M$ on the base space~$M$,
related to one another by the bundle projection: then equations~%
(\ref{eq:VFREPVAR1}) and~(\ref{eq:VFREPVAR2}) should be written as
\begin{equation} \label{eq:VFREPVAR3} 
 \delta_X^{} \phi~=~X_P^{}(\phi) \, - \, T\phi \, (X_M^{})~,
\end{equation}
and
\begin{equation} \label{eq:VFREPVAR4} 
 \delta_X^{} \phi(m)~=~X_P^{}(\phi(m)) \, - \, T_m^{} \phi \, (X_M^{}(m))
 \qquad \mbox{for $\, m \in M$}~,
\end{equation}
respectively.
The same argument as in the previous paragraph can then be used
to prove the following
\begin{lem}~ \label{lem:VFREPVAR} 
 Let\/ $\phi$ be a section of a fiber bundle\/~$P$ over a base
 manifold\/~$M$.
 Given any vector field\/~$X_M^{}$ on\/~$M$, there exists
 a projectable vector field\/~$X_P^\phi$ on\/~$P$ which is\/
 $\phi$-related to\/~$X_M^{}$, i.e., satisfies $\; X_P^\phi(\phi)
 = T\phi(X_M)$, and then we have $\; \phi^*(i_{X_P^\phi} \alpha) =
 i_{X_M^{}}^{}(\phi^* \alpha)$, for any differential form\/ $\alpha$
 on\/~$P$.
\end{lem}

As an example of how useful the representation of variations of sections
of a fiber bundle by composition with vertical vector fields or even
projectable vector fields can be, we present explicit formulas for the
first and second functional derivative of a local functional of the type
considered above~-- beginning with a more detailed definition of this
class of functionals.
\begin{dfn} \label{def:LFUNCT}~ 
 Given a fiber bundle\/ $P$ over an\/ $n$-dimensional base manifold\/~$M$,
 let\/ $\Sigma$ be a\/ $p$-dimensional submanifold of\/~$M$, possibly with
 boundary\/~$\partial\Sigma$, and\/ $f$ be a\/ $p$-form on the total
 space\/~$P$ such that the intersection of\/ $\Sigma$ with the base
 support of\/~$f$ is compact.
 The \textbf{local functional} associated to\/ $\Sigma$ and\/~$f$ is
 the functional $\; \mathcal{F}_{\Sigma,f}^{}: \mathscr{C} \longrightarrow
 \mathbb{R} \,$ on the space $\, \mathscr{C} \subset \Gamma^\infty(P)$
 \linebreak of field configurations defined by
 \begin{equation} \label{eq:func01}
  \mathcal{F}_{\Sigma,f}^{}[\phi]~=~\int_\Sigma \phi^* f
  \qquad \mbox{for $\, \phi \in \mathscr{C}$}~.
 \end{equation}
\end{dfn}
These functionals are differentiable, and their derivative is given by
a completely explicit formula:
\begin{prp} \label{prop:FFUNDER}~ 
 Given a fiber bundle\/ $P$ over an\/ $n$-dimensional base manifold\/~$M$,
 let\/ $\Sigma$ be a\/ $p$-dimensional submanifold of\/~$M$, possibly with
 boundary\/~$\partial\Sigma$, and\/ $f$ be a\/ $p$-form on the total
 space\/~$P$ such that the intersection of\/ $\Sigma$ with the base
 support of\/~$f$ is compact.
 Then the local functional\/ $\mathcal{F}_{\Sigma,f}^{}$ associated to\/
 $\Sigma$ and\/~$f$ is differentiable, and representing variations of
 sections of\/~$P$ in the form\/ $\delta_X \phi$ where $\, X = (X_P^{},%
 X_M^{}) \,$ is a projectable vector field, its functional derivative
 is given by the formula
 \begin{equation} \label{eq:FFUNDER1} 
  \mathcal{F}_{\Sigma,f}^{\>\prime}[\phi] \cdot \delta_X^{} \phi~
  =~\int_\Sigma \Big( \phi^* \big( L_{X_P}^{} f \big) \, - \,
                     L_{X_M}^{} \big( \phi^* f \big) \Big)
  \qquad \mbox{for $\, \phi \in \mathscr{C}$,
                   $\delta_X^{} \phi \in T_\phi^{} \mathscr{C}$}~,
 \end{equation}
 where\/ $L_Z$ denotes the Lie derivative along the vector field\/~$Z$.
\end{prp}
\begin{rmk}~
 Under the boundary condition that the intersection of\/~$\partial\Sigma$
 with the base support of\/~$f$ is empty, equation~(\ref{eq:FFUNDER1}) can
 be rewritten as follows:
 \begin{equation} \label{eq:FFUNDER2} 
   \mathcal{F}_{\Sigma,f}^{\>\prime}[\phi] \cdot \delta_X^{} \phi~
  =~\int_\Sigma \Big( \phi^* \big( i_{X_P}^{} df \big) \, - \,
                     i_{X_M}^{} \big( \phi^* df \big) \Big)
  \qquad \mbox{for $\, \phi \in \mathscr{C}$,
                   $\delta_X^{} \phi \in T_\phi^{} \mathscr{C}$}~.
 \end{equation}
 The same equation holds when this boundary condition is replaced by the
 requirement that\/~$\delta_X^{} \phi$ should vanish on\/~$\partial\Sigma$.
\end{rmk}
\begin{proof}
 Recall first that for any functional $\mathslf{F}$ on~$\mathscr{C}$,
 its functional derivative at $\, \phi \in \mathscr{C} \,$ along
 $\, \delta\phi \in T_\phi^{} \mathscr{C} \,$ is defined by
 \[
  \mathslf{F}^{\>\prime}[\phi] \cdot \delta\phi~
  =~\frac{d}{d \lambda} \, \mathslf{F} \, [\phi_\lambda^{}] \!\>
    \Big|_{\lambda=0}~,
 \]
 where the $\, \phi_\lambda^{} \in \mathscr{C} \,$ constitute a smooth
 one-parameter family of sections of~$P$ such that
 \[
  \phi~=~\phi_\lambda^{} \!\> \Big|_{\lambda=0}~~~,~~~
  \delta\phi~=~\frac{d}{d\lambda} \, \phi_\lambda^{} \!\> \Big|_{\lambda=0}~.
 \]
 Fixing $\phi$ and $\delta\phi$ and choosing a projectable vector field
 $\, X = (X_P^{},X_{M}^{}) \,$ that represents $\delta\phi$ as $\delta_X^{}
 \phi$, according to equation~(\ref{eq:VFREPVAR1}), consider its flow, which
 is a (local) one-parameter group of (local) automorphisms $\, \Phi_\lambda^{}
 = (\Phi_{P,\lambda}^{},\Phi_{M,\lambda}^{}) \,$ such that
 \[
  X_P^{}~=~\frac{d}{d\lambda} \, \Phi_{P,\lambda}^{} \!\> \Big|_{\lambda=0}
  ~~~\mbox{and}~~~
  X_M^{}~=~\frac{d}{d\lambda} \, \Phi_{M,\lambda}^{} \!\> \Big|_{\lambda=0}~.
 \]
 This allows us to take the one-parameter family of sections
 $\phi_\lambda^{}$ of~$P$ to be given by the one-parameter group
 of automorphisms $\Phi_\lambda^{}$, according to
 \[
  \phi_\lambda^{}~
  =~\Phi_{P,\lambda}^{} \vcirc \phi \vcirc \, \Phi_{M,\lambda}^{-1}~.
 \]
 Now we are ready to calculate:
 \begin{eqnarray*}
  \mathcal{F}_{\Sigma,f}^{\>\prime}[\phi] \cdot \delta_X^{} \phi \!
  &=&\!\! \frac{d}{d\lambda} \; \bigg( \int_{\Sigma} \phi_\lambda^* f \bigg)
          \bigg|_{\lambda=0} \\[2mm]
  &=&\!\! \int_\Sigma \; \frac{d}{d\lambda} \; \phi_\lambda^* f \,
          \bigg|_{\lambda=0}~
   =~     \int_\Sigma \; \frac{d}{d\lambda} \;
          \big( \Phi_{P,\lambda}^{} \vcirc \phi \vcirc \, \Phi_{M,\lambda}^{-1}
                \big)^* f \, \bigg|_{\lambda=0} \\[2mm]
  &=&\!\! \int_\Sigma \; \bigg( \phi_{}^*
          \bigg( \frac{d}{d\lambda} \;
                 \big( \Phi_{P,\lambda}^{\phantom{P,\lambda} *} f \big)
                 \bigg|_{\lambda=0} \bigg) \, + \,
          \frac{d}{d\lambda}
          \Big( \Phi_{M,\lambda}^{-1~\,*} \big( \phi^*f \big) \Big)
          \bigg|_{\lambda=0} \bigg) \\[2mm]
  &=&\!\! \int_\Sigma \Big( \phi^* \big( L_{X_P}^{} f \big) \, - \,
                           L_{X_M}^{} \big( \phi_{}^* f \big) \Big)~.
 \end{eqnarray*}
 To derive equation~(\ref{eq:FFUNDER2}) from equation~(\ref{eq:FFUNDER1}), it
 suffices to apply standard formulas such as $\, L_Z = d \, i_Z + i_Z \, d \,$
 and the fact that $d$ commutes with pull-backs, together with Stokes'
 theorem, and use the boundary condition~(\ref{eq:BOUNDC}) to kill the
 resulting two integrals over $\partial\Sigma$.
 The same argument works when $\delta_X^{} \phi$ is supposed to vanish
 on~$\partial\Sigma$, since we may then arrange $X_M^{}$ to vanish on~%
 $\partial\Sigma$ and $X_P^{}$ to vanish on~$P|_{\partial\Sigma}$.
\qed
\end{proof}

\section{Multiphase spaces and multisymplectic structure}

As has already been stated before, the bundle~$P$ appearing in the
previous two sections, representing the \emph{multiphase space} of
the system under consideration, will be required to carry additional
structure, namely a \emph{multisymplectic form}.
There has been much debate and even some confusion in the literature
on what should be the ``right'' definition of the concept of a multi%
symplectic structure, but all proposals made so far can be subsumed
under the following
\begin{dfn} \label{dfn:MSPF}~
 A \textbf{multisymplectic form} on a manifold~$P$ is a differential
 form $\,\omega$ on\/~$P$ of degree\/ $n\!+\!1$, say, which is closed,
 \begin{equation} \label{eq:MSFC} 
  d \>\! \omega = 0~,
 \end{equation}
 and satisfies certain additional algebraic constraints~-- among
 them that of being non-degenerate, in the sense that for any
 vector field\/~$X$ on~$P$, we have
 \begin{equation}
  i_X^{} \omega = 0~~\Longrightarrow~~X = 0~.
 \end{equation}
\end{dfn}
Of course, this definition is somewhat vague since it leaves open what
other algebraic constraints should be imposed besides non-degeneracy.
One rather natural criterion is that they should be sufficient to
guarantee the validity of a Darboux type theorem.
Clearly, when $n\!=\!1$, the above definition reduces to that of
a symplectic form, and no additional constraints are needed.
But when $n\!>\!1$, which is the case pertaining to field theory
rather than to mechanics, this is no longer so.
An important aspect here is that $P$ is not simply a manifold
but rather the total space of a fiber bundle over the space-time
manifold~$M$, which is supposed to be $n$-dimensional, so one
restriction is that the degree of the form~$\omega$ is linked
to the space-time dimension.
Another restriction is that $\omega$ should be $(n\!-\!1)$-horizontal,
which means that its contraction with three vertical vector fields
vanishes:
\begin{equation}
 i_X^{} i_Y^{} i_Z^{} \>\! \omega = 0 \qquad \mbox{for~$\, X,Y,Z$ vertical}~.
\end{equation}
And finally, there is a restriction that, roughly speaking, guarantees
existence of a ``sufficiently high-dimensional'' lagrangian subbundle
of the tangent bundle of~$P$, but since we shall not need it here,
we omit the details: they can be found in Ref.~\cite{FG}.

The main advantage of the definition of a multisymplectic form as
given above is that we can proceed to discuss a number of concepts 
which do not depend on the precise nature of the additional algebraic
constraints to be imposed.
For example, a vector field $X$ on~$P$ is said to be \emph{\textbf%
{locally hamiltonian}} if $i_X^{} \omega$ is closed, that is, if\,%
\footnote{Throughout this paper, we shall make use of Cartan's formula
$\, L_X^{} = d \, i_X^{} + i_X^{} d \,$ without further mention.}
\begin{equation} \label{eq:LHVF}
 L_X^{} \omega = 0~,
\end{equation}
and is said to be \emph{\textbf{globally hamiltonian}} or simply
\emph{\textbf{hamiltonian}} if $i_X^{} \omega$ is exact, that is,
if there exists an $(n\!-\!1)$-form $f$ on~$P$ such that
\begin{equation} \label{eq:GHVF}
 i_X^{} \omega = df~.
\end{equation}
Reciprocally, an $(n\!-\!1)$-form $f$ on~$P$ is said to be \emph%
{\textbf{hamiltonian}} if there exists a vector field~$X$ on~$P$
such that equation~(\ref{eq:GHVF}) holds: this condition is
trivially satisfied when $n\!=\!1$ but not when~$n\!>\!1$.
Note that due to non-degeneracy of~$\omega$, $X$ is uniquely determined
by~$f$, and will therefore often denoted by~$X_f$, whereas $f$ is
determined by~$X$ only up to addition of a closed form: despite this
(partial) ambiguity, we shall say that $X$ is associated with~$f$ and
$f$ is associated with~$X$.
In the special case when $\omega$ is exact, i.e., we have
\begin{equation} \label{eq:MSFE} 
 \omega = - \, d \>\! \theta~,
\end{equation}
where $\theta$ is an appropriate $n$-form on~$P$ called the \emph{%
\textbf{multicanonical form}}, a vector field $X$ on~$P$ is said to
be \emph{\textbf{exact hamiltonian}} if
\begin{equation} \label{eq:EHVF}
 L_X^{} \theta = 0~.
\end{equation}
In this case, of course, the associated hamiltonian form can be
simply chosen to be $i_X^{} \theta$, since $\; d \>\! i_X^{} \theta
= L_X^{} \theta \, - \, i_X^{} d \>\! \theta = i_X^{} \omega$.
In particular, this happens when $P$ is the total space of a
vector bundle (over some base space~$E$, say, which in turn
will be the total space of a fiber bundle over the space-time
manifold~$M$), provided that $\omega$ is homogeneous of degree
one with respect to the corresponding \emph{\textbf{Euler vector
field}} or  \emph{\textbf{scaling vector field}} $\Sigma$, i.e.,
\begin{equation} \label{eq:MSFH} 
 L_\Sigma^{} \omega = \omega~,
\end{equation}
since we may then define $\theta$ by
\begin{equation} \label{eq:HMSFE} 
 \theta = - \, i_\Sigma^{} \omega~.
\end{equation}
Moreover, we can then employ $\Sigma$ to decompose vector fields and
differential forms on~$P$ according to their scaling degree, and as
we shall see below, this turns out to be extremely useful for the
classification of hamiltonian vector fields (whether locally or
globally or exact) and of hamiltonian forms.

The standard example of this kind of structure is provided by any (first
order) hamiltonian system obtained from a (first order) lagrangian system
via a \emph{non-degenerate} covariant Legendre transformation.
In this approach, one starts out from another fiber bundle over~$M$,
denoted here by~$E$ and called the \emph{configuration bundle}: the
relation between the two bundles $E$ and~$P$ is then established
by taking appropriate duals of first order jet bundles.
Namely, consider the first order jet bundle of~$E$, denoted simply
by $JE$, which is both a fiber bundle over~$M$ (with respect to the
source projection) and an affine bundle over~$E$ (with respect to
the target projection), together with its difference vector bundle,
called the linearized first order jet bundle of~$E$ and denoted
here by $\vec{J} E$, which is both a fiber bundle over~$M$ (with
respect to the source projection) and a vector bundle over~$E$
(with respect to the target projection), and introduce the
corresponding duals: the affine dual $J^{\circledstar} E$ of~$JE$ and
the usual linear dual $\vec{J}^{\circledast} E$ of~$\vec{J} E$. \linebreak
In what follows, we shall refer to the latter as the \emph{ordinary
multiphase space}\/ and to the former as the \emph{extended multiphase
space}\/ of the theory.
As it turns out and has been emphasized since the beginning of the
``modern phase'' of the development of the subject in the early 1990s
(see, e.g., \cite{CCI}), both of these play an important role since
not only are both of them fiber bundles over~$M$ and vector bundles
over~$E$, but $J^{\circledstar} E$ is also an affine line bundle over~%
$\vec{J}^{\circledast} E$, and the dynamics of the theory is given by the
choice of a \emph{hamiltonian}\/, which is a section $\, \mathcal{H}:
\vec{J}^{\circledast} E \longrightarrow J^{\circledstar} E \,$ of this affine
line bundle.
Moreover, and this is of central importance, both of these
multiphase spaces carry a multisymplectic structure.
Namely, $J^{\circledstar} E$ comes with a naturally defined \emph%
{multisymplectic form}\/ of degree~$n\!\!\>+\!1$, denoted here
by $\omega$, which (up to sign) is the exterior derivative of
an equally naturally defined \emph{multicanonical form}\/ of
degree~$n$, denoted here by $\theta\,$: $\omega = - d \>\! \theta$,
and if we choose a hamiltonian $\, \mathcal{H}: \vec{J}^{\circledast} E
\longrightarrow J^{\circledstar} E$, we can pull them back to obtain
a corresponding multicanonical form
\begin{equation}
 \theta_{\mathcal{H}}^{\vphantom{j}}~=~\mathcal{H}^* \theta
\end{equation}
and a corresponding multisymplectic form
\begin{equation}
 \omega_{\mathcal{H}}^{\vphantom{j}}~=~\mathcal{H}^* \omega
\end{equation}
on~$\vec{J}^{\circledast} E$: again,
$\omega_{\mathcal{H}} = - d \>\! \theta_{\mathcal{H}}$.
Thus the main difference between the extended and the ordinary
multiphase space is that $\theta$ and~$\omega$ are universal
and ``cinematical'', whereas $\theta_{\mathcal{H}}$ and~%
$\omega_{\mathcal{H}}$ are ``dynamical''.
In terms of local Darboux coordinates $(x^\mu,q^i,p_i^\mu)$
for~$\vec{J}^{\circledast} E$ and $(x^\mu,q^i,p_i^\mu,p)$ for~%
$J^{\circledstar} E$ induced by the choice of local coordinates
$x^\mu$ for~$M$, $q^i$ for the typical fiber $Q$ of~$E$ and a
local trivialization of~$E$ over~$M$, we have
\begin{equation} \label{eq:MCANF1}
 \theta~=~p\>\!_i^\mu \; dq_{}^i \>\smwedge\, d^{\,n} x_\mu^{} \,
          + \, p \; d^{\,n} x~,
\end{equation}
and
\begin{equation} \label{eq:MSPF01}
 \omega~
 =~dq_{}^i \>\smwedge\, dp\>\!_i^\mu \>\smwedge\, d^{\,n} x_\mu^{} \, - \,
   dp \;\smwedge\, d^{\,n} x~,
\end{equation}
so that writing $\, \mathcal{H} = - \, H \, d^{\,n} x$,
\begin{equation} \label{eq:MCANF2}
 \theta_{\mathcal{H}}^{}~
 =~p\>\!^\mu_i \; dq^i \,\smwedge\, d^{\,n} x_\mu \, - \, H \, d^{\,n} x~,
\end{equation}
and
\begin{equation} \label{eq:MSPF2a}
 \omega_{\mathcal{H}}^{}~
 =~dq^i \,\smwedge\, dp\>\!^\mu_i \,\smwedge\, d^{\,n} x_\mu \, + \,
   dH \,\smwedge\, d^{\,n} x~,
\end{equation}
or more explicitly
\begin{equation} \label{eq:MSPF2b}
 \omega_{\mathcal{H}}^{}~
 =~dq^i \,\smwedge\, dp\>\!^\mu_i \,\smwedge\, d^{\,n} x_\mu \, + \,
   \frac{\partial H}{\partial q^i} \; dq^i \,\smwedge\, d^{\,n} x \, + \,
   \frac{\partial H}{\partial p\>\!_i^\mu} \; dp\>\!_i^\mu \,\smwedge\, d^{\,n} x~.
\end{equation}
where $d^{\,n} x_\mu^{}$ is the (local) $(n\!-\!1)$-form obtained by
contracting the (local) volume form $d^{\,n} x$ with the local vector
field $\, \partial_\mu^{} \equiv \partial/\partial x_{}^\mu$; for more
details, including a global definition of~$\theta$ that does not depend
on any of these choices, we refer to~\cite{CCI,FSR,GIM}.%
\footnote{It should be noted that whereas the form $\omega$ is always
non-degenerate, the form $\omega_{\mathcal{H}}^{}$ is degenerate in
mechanics ($n=1$) and non-degenerate in field theory ($n>1$).}

\section{Variational principle and equations of motion}

The fundamental link that merges the functional and multisymplectic
formalisms discussed in the previous two sections into one common
picture becomes apparent when the construction of functionals of fields
from forms on multiphase space outlined in the introduction is applied
to the multicanonical $n$-form $\theta_{\mathcal{H}}^{}$ on (ordinary)
multiphase space: this provides the \emph{action functional}\/
$\mathslf{S}$ of the theory, defining the variational principle
whose stationary points are the solutions of the equations of motion.
Indeed, the action functional is really an entire family of functionals
$\mathslf{S}_K^{}$ on the space $\mathscr{C}$ of field configurations
$\phi$ (see equation~(\ref{eq:FCONFS1})), given by
\begin{equation} \label{eq:ACTION1}
 \mathslf{S}_K^{}[\phi]~
 =~\int_K \phi^\ast \, \theta_{\mathcal{H}}^{}~,
\end{equation}
where $K$ runs through the compact submanifolds of~$M$ which are the
closure of their interior in~$M$ and have smooth boundary $\partial K$.

Within this setup, a section $\phi$ in~$\mathscr{C}$ is said to be a
\emph{stationary point of the action}\/ if, for any compact submanifold
$K$ of~$M$ which is the closure of its interior in~$M$ and has smooth
boundary $\partial K$, $\phi$ becomes a critical point of the functional
$\mathslf{S}_K^{}$ restricted to the (formal) submanifold
\begin{equation}
 \mathscr{C}_{K,\phi}^{}~
 =~\big\{ \tilde{\phi} \in \mathscr{C}~|~
          \tilde{\phi}|_{\partial K}^{} = \phi|_{\partial K}^{} \bigr\}
\end{equation}
of~$\mathscr{C}$, or equivalently, if the functional derivative
$\mathslf{S}_K^{\>\prime}[\phi]$ of $\mathslf{S}_K^{}$ at~$\phi$
vanishes on the subspace
\begin{equation} \label{eq:FTANS2}
 T_\phi^{} \mathscr{C}_{K,\phi}^{}~
 =~\big\{ \delta\phi \in T_\phi^{} \mathscr{C}~|~
          \delta\phi = 0~\mbox{on}~\partial K \big\}
\end{equation}
of~$T_\phi^{} \mathscr{C}$.
As is well known, this is the case if and only if $\phi$ satisfies
the corresponding equations of motion, which in the present case are
the De\,Donder\,--\,Weyl equations; see, e.g., \cite{FSR,Got,GIM}.
Globally, these can be cast in the form
\begin{equation} \label{eq:FEQMOT1} 
 \begin{array}{c}
  \phi^* (i_X^{} \, \omega_{\mathcal{H}}^{})~=~0 \\[1mm]
  \mbox{for any vertical vector field $X$ on~$P$}
 \end{array}~,
\end{equation}
or even
\begin{equation} \label{eq:FEQMOT2} 
 \begin{array}{c}
  \phi^* \big( i_{X_P}^{} \omega_{\mathcal{H}}^{\vphantom{j}} \big)~=~0 \\[1mm]
  \mbox{for any projectable vector field $\, X = (X_P^{},X_M^{})$}
 \end{array}~,
\end{equation}
whereas, when written in terms of local Darboux coordinates
$(x^\mu,q^i,p_i^\mu)$ as before, they read
\begin{equation} \label{eq:FDDW1}
 \partial_\mu \varphi^i~
 =~\frac{\partial H}{\partial p_i^\mu}(\varphi,\pi)~~~,~~~
 \partial_\mu \pi_i^\mu~
 = \; - \, \frac{\partial H}{\partial q^i}(\varphi,\pi)~,
\end{equation}
where $\, P = \vec{J}^{\circledast} E$, $\phi = (\varphi,\pi) \,$ and
$\, \mathcal{H} = - H \, d^{\,n} x$.
Similarly, given such a solution~$\phi$, we shall say that a section
$\delta\phi$ in $T_\phi^{} \mathscr{C}$ is an \emph{infinitesimal
stationary point of the action}\/ if it is formally tangent to
the (formal) submanifold of solutions, or equivalently, if it
satisfies the corresponding linearized equations of motion, which
in the present case are the linearized De\,Donder\,--\,Weyl equations.
Globally, representing $\delta\phi$ in the form $\delta_X \phi$ where
$\, X = (X_P^{},X_M^{}) \,$ is a projectable vector field, these can be
cast in the form
\begin{equation} \label{eq:LEQMOT1}
 \begin{array}{c}
  \phi^* (i_Y^{} L_{X_P}^{} \omega_{\mathcal{H}}^{\vphantom{j}})~=~0 \\[1mm]
  \mbox{for any vertical vector field $Y$ on~$P$}
 \end{array}~,
\end{equation}
or even
\begin{equation} \label{eq:LEQMOT2}
 \begin{array}{c}
  \phi^* (i_{Y_P}^{} L_{X_P}^{} \omega_{\mathcal{H}}^{\vphantom{j}})~=~0 \\[1mm]
  \mbox{for any projectable vector field $\, Y = (Y_P^{},Y_M^{})$}
 \end{array}~,
\end{equation}
whereas, when written in terms of local Darboux coordinates
$(x^\mu,q^i,p_i^\mu)$ as before, they read
\begin{equation} \label{eq:LDDW1}
 \begin{array}{c}
  \partial_\mu \, \delta \varphi^i~
  =~{\displaystyle \mbox{} + \,
     \frac{\partial^{\>\!2} H}{\partial q^j \, \partial p\>\!_i^\mu}
     (\varphi,\pi) \; \delta \varphi^j \, + \,
     \frac{\partial^{\>\!2} H}{\partial p\>\!_j^\nu \, \partial p\>\!_i^\mu} \,
     (\varphi,\pi) \; \delta \pi_j^\nu}~, \\[3ex]
  \partial_\mu \, \delta \pi_i^\mu~
  =~{\displaystyle \mbox{} - \,
     \frac{\partial^{\>\!2} H}{\partial q^j \, \partial q^i}
     (\varphi,\pi) \; \delta \varphi^j \, - \,
     \frac{\partial^{\>\!2} H}{\partial p\>\!_j^\nu \, \partial q^i}
     (\varphi,\pi) \; \delta \pi_j^\nu}~,
 \end{array}
\end{equation}
where $\, P = \vec{J}^{\circledast} E$, $\phi = (\varphi,\pi)$, $\delta\phi
= (\delta\varphi,\delta\pi) \,$ and $\, \mathcal{H} = - H \, d^{\,n} x$.

\begin{proof}
 For the first part (concerning the derivation of the full equations
 of motion from the variational principle), we begin by specializing
 equation~(\ref{eq:FFUNDER1}) to vertical vector fields $X$ on~$P$
 (i.e., setting $\, X_M^{} = 0 \,$ and replacing $X_P^{}$ by~$X$),
 with $\, f = \theta_{\mathcal{H}}$, and using standard facts such as
 the formula $\, L_Z = d \, i_Z + i_Z \, d \,$ or that $d$ commutes
 with Lie derivatives and pull-backs, together with Stokes' theorem,
 to obtain that, for any vertical vector field~$X$ on~$P$,
 \[
  \mathcal{S}_K^{\>\prime}[\phi] \cdot \delta_X^{} \phi~
  = \; - \, \int_K \phi^* ( i_X^{} \, \omega_{\mathcal{H}}^{} \big)
    \, + \, \int_{\partial K} \phi^*
            ( i_X^{} \, \theta_{\mathcal{H}}^{} \big)~.
 \]
 Obviously, condition~(\ref{eq:FEQMOT1}) implies that this expression will
 be equal to zero for all vertical vector fields~$X$ on~$P$ which vanish
 on~$P|_{\partial K}$. Conversely, it follows from Lemma~\ref{lem:DENS}
 below that if this is the case, then condition~(\ref{eq:FEQMOT1}) holds.
 Moreover, it is easily seen that this condition is really equivalent to
 the condition
 \[
 \begin{array}{c}
  \phi^* \big( i_{X_P}^{} \omega_{\mathcal{H}}^{\vphantom{j}} \big) \, - \,
  i_{X_M}^{} \big( \phi^* \omega_{\mathcal{H}}^{\vphantom{j}} \big)~=~0 \\[1mm]
  \mbox{for any projectable vector field $\, X = (X_P^{},X_M^{})$}
 \end{array}~,
 \]
 but it so happens that the form $\phi^* \omega_{\mathcal{H}}^{}$ is
 identically zero, for dimensional reasons.
 Finally, a simple calculation shows that equation~(\ref{eq:FEQMOT1}),
 when written out explicitly in local Darboux coordinates, can be
 reduced to the system~(\ref{eq:FDDW1}). 
 \\[1mm]
 For the second part (concerning the linearization of the full equations
 of motion around a given solution~$\phi$), we proceed as in the proof of
 Proposition~\ref{prop:FFUNDER}: fixing~$\phi$ and~$\delta\phi$, suppose
 we are given a smooth one-parameter family of sections $\, \phi_\lambda^{}
 \in \mathscr{C} \,$ of~$P$ such that
 \[
  \phi~=~\phi_\lambda^{} \!\> \Big|_{\lambda=0}~~~,~~~
  \delta\phi~=~\frac{d}{d\lambda} \, \phi_\lambda^{} \!\> \Big|_{\lambda=0}~,
 \]
 as well as a projectable vector field $\, X = (X_P^{},X_{M}^{}) \,$ that
 represents $\delta\phi$ as $\delta_X^{} \phi$, according to equation~%
 (\ref{eq:VFREPVAR1}), together with its flow, which is a (local)
 one-parameter group of (local) automorphisms $\, \Phi_\lambda^{}
 = (\Phi_{P,\lambda}^{},\Phi_{M,\lambda}^{}) \,$ such that
 \[
  X_P^{}~=~\frac{d}{d\lambda} \, \Phi_{P,\lambda}^{} \!\> \Big|_{\lambda=0}
  ~~~\mbox{and}~~~
  X_M^{}~=~\frac{d}{d\lambda} \, \Phi_{M,\lambda}^{} \!\> \Big|_{\lambda=0}~,
 \]
 allowing us to take the one-parameter family of sections $\phi_\lambda^{}$
 of~$P$ to be given by the one-parameter group of automorphisms
 $\Phi_\lambda^{}$, according to
 \[
  \phi_\lambda^{}~
  =~\Phi_{P,\lambda}^{} \vcirc \phi \vcirc \, \Phi_{M,\lambda}^{-1}~.
 \]
 Then for any vertical vector field $Y$ on~$P$, we have
 \begin{eqnarray*}
  \frac{d}{d\lambda} \, \phi_\lambda^*
  \bigl( i_Y^{} \, \omega_{\mathcal{H}} \bigr) \Big|_{\lambda=0} \!\!
  &=&\!\! \frac{d}{d\lambda} \, \Bigl( \Phi_{P,\lambda}^{} \vcirc \phi \vcirc
                                       \Phi_{M,\lambda}^{-1} \Bigr)^{\!\!*}
          \bigl( i_Y^{} \, \omega_{\mathcal{H}} \bigr) \Big|_{\lambda=0} \\[1ex]
  &=&\!\! \frac{d}{d\lambda} \, \phi^* \Bigl( \Phi_{P,\lambda}^*
          (i_Y^{} \, \omega_{\mathcal{H}}) \Bigr) \Big|_{\lambda=0} \, + \,
          \frac{d}{d\lambda} \, (\Phi_{M,\lambda}^{-1})^* \Bigl( \phi^*
          (i_Y^{} \, \omega_{\mathcal{H}}) \Bigr) \Big|_{\lambda=0} \\[2ex]
  &=&\!\! \phi^* \bigl( L_{X_P}^{} i_Y^{} \, \omega_{\mathcal{H}} \bigr) \, - \,
          L_{X_M}^{} \bigl( \phi^* (i_Y^{} \, \omega_{\mathcal{H}}) \bigr)
  \\[2ex]
  &=&\!\! \phi^* \bigl( i_Y^{} L_{X_P}^{} \omega_{\mathcal{H}} \bigr) \, + \,
          \phi^* \bigl( i_{[X_P,Y]}^{} \, \omega_{\mathcal{H}} \bigr) \, - \,
          L_{X_M}^{} \bigl( \phi^* (i_Y^{} \, \omega_{\mathcal{H}}) \bigr)~,
 \end{eqnarray*}
 where the last two terms vanish according to equation~(\ref{eq:FEQMOT1}),
 and the same argument holds with $Y$ replaced by~$Y_P^{}$ where
 $\, Y = (Y_P^{},Y_{M}^{}) \,$ is any projectable vector field.
 Finally, an elementary but somewhat lengthy calculation shows that
 equation~(\ref{eq:LEQMOT1}), when written out explicitly in local
 Darboux coordinates, can be reduced to the system~(\ref{eq:LDDW1}),
 provided $\phi$ satisfies the system~(\ref{eq:FDDW1}). \\
 \hspace*{\fill} \qed
\end{proof}

\noindent
The lemma we have used in the course of the argument is the following.
\begin{lem} \label{lem:DENS}~
 Given a fiber bundle\/ $P$ over an\/ $n$-dimensional base manifold\/~$M$,
 let\/ $\phi$ be a section of\/~$P$ and\/ $\alpha$ be an\/ $(n+1)$-form
 on\/~$P$ such that, for any compact submanifold\/ $K$ of~$M$ which is the
 closure of its interior in~$M$ and has smooth boundary\/~$\partial K$,
 and for any vertical vector field\/ $X$ on\/~$P$ that vanishes on\/~%
 $P|_{\partial K}$ together with all its derivatives, the integral
 \[
  \int_K \phi^\ast (i_X^{} \alpha)
 \]
 vanishes. Then the form $\, \phi^\ast (i_X^{} \alpha) \,$ vanishes
 identically, for any vertical vector field $X$ on\/~$P$ (not subject
 to any boundary conditions).
\end{lem}
\begin{proof}
 Suppose $X$ is any vertical vector field on~$P$ and $m$ is a point
 in~$M$ where $\, \phi^\ast (i_X^{} \alpha) \,$ does not  vanish.
 Choosing an appropriately oriented system of local coordinates
 $x^\mu$ in~$M$ around~$m$, we may write $\, \phi^\ast \, (i_X^{} \alpha)
 = a \; d^{\,n} x \,$ where $a$ is a function that is strictly  positive
 at the coordinate origin (corresponding to the point~$m$) and hence, for
 an appropriate choice of sufficiently small positive numbers $\delta$
 and~$\epsilon$, will be $\geqslant \epsilon$ on~$\bar{B}_\delta$
 (here we denote by $B_r$ and by $\bar{B}_r$ the open and closed
 ball of radius $r$ around the coordinate origin, respectively).
 Choosing a test function $\chi$ on~$M$ such that $\, 0 \leqslant
 \chi \leqslant 1$, $\mathrm{supp} \, \chi \subset B_\delta \,$ and
 $\, \chi = 1$ on~$\bar{B}_{\delta/2}$, lifted to~$P$ by pull-back
 and multiplied by~$X$ to give a new vertical vector field $\chi X$
 on~$P$ that vanishes on~$P|_{\partial \bar{B}_\delta}$ together with all
 its derivatives, we get
 \[
  \int_{\bar{B}_\delta} \phi^\ast (i_{\chi X}^{} \alpha)~
  =~\int_{\bar{B}_\delta} d^{\,n} x~\chi(x) a(x)~
  \geqslant~\int_{\bar{B}_{\delta/2}} d^{\,n} x~a(x)~
  ~\geqslant~\epsilon~\mathrm{vol}(\bar{B}_{\delta/2})~>~0~,
 \]
 which is a contradiction.
 \qed
\end{proof}

To summarize, the space of stationary points of the action, denoted here
by~$\mathscr{S}$ and considered as a (formal) submanifold of the space
of field configurations~$\mathscr{C}$, can be described in several
equivalent ways: we have
\begin{equation} \label{eq:COVPHS1}
 \mathscr{S}~=~\big\{ \phi \in \mathscr{C}~|~\mbox{$\phi$ is
                      stationary point of the action} \big\}~,
\end{equation}
or
\begin{equation} \label{eq:COVPHS2}
 \mathscr{S}~
 =~\big\{ \phi \in \mathscr{C}~|~\mbox{$\phi$
          satisfies the equations of motion (\ref{eq:FDDW1})} \big\}~,
\end{equation}
or
\begin{equation} \label{eq:COVPHS3}
 \mathscr{S}~
 =~\big\{ \phi \in \mathscr{C}~|~
          \phi^* (i_{X_P}^{} \omega_{\mathcal{H}}^{}) = 0~~
          \mbox{for any projectable vector field $\, X = (X_P^{},X_M^{})$}
          \big\}~,
\end{equation}
or
\begin{equation} \label{eq:COVPHS4}
 \mathscr{S}~
 =~\big\{ \phi \in \mathscr{C}~|~
          \phi^* (i_X^{} \, \omega_{\mathcal{H}}^{}) = 0~~
          \mbox{for any vertical vector field $X$ on~$P$}
          \big\}~.
\end{equation}

\noindent
This space $\mathscr{S}$ plays a central role in the functional approach:
it is widely known under the name of \emph{covariant phase space}.
Moreover, given $\phi$ in~$\mathscr{S}$ and allowing $\, X = (X_P^{},%
X_M^{}) \,$ to run through the projectable vector fields, the (formal)
tangent space $T_\phi^{} \mathscr{S}$ to~$\mathscr{S}$ at~$\phi$ is
\begin{equation} \label{eq:TSCPHS1}
 T_\phi^{} \mathscr{S}~
 =~\big\{ \delta\phi \in T_\phi^{} \mathscr{C}~|~\mbox{$\delta\phi$ is
          infinitesimal stationary point of the action} \big\}~,
\end{equation}
or
\begin{equation} \label{eq:TSCPHS2}
 T_\phi^{} \mathscr{S}~
 =~\big\{ \delta\phi \in T_\phi^{} \mathscr{C}~|~\mbox{$\delta\phi$ satisfies
           the linearized equations of motion (\ref{eq:LDDW1})} \big\}~,
\end{equation}
or
\begin{equation} \label{eq:TSCPHS3}
 T_\phi^{} \mathscr{S}~
 =~\big\{ \delta_X^{} \phi \in T_\phi^{} \mathscr{C}~| \;
          \parbox{7.8cm}{\begin{center}
           $\phi^* (i_{Y_P}^{} L_{X_P}^{} \omega_{\mathcal{H}}^{}) = 0$ \\[1mm]
           \mbox{for any projectable vector field $\, Y = (Y_P^{},Y_M^{})$}
          \end{center}}
          \big\}~,
\end{equation}
or
\begin{equation} \label{eq:TSCPHS4}
 T_\phi^{} \mathscr{S}~
 =~\big\{ \delta_X^{} \phi \in T_\phi^{} \mathscr{C}~| \;
          \parbox{6.1cm}{\begin{center}
           $\phi^* (i_Y^{} L_{X_P}^{} \omega_{\mathcal{H}}^{}) = 0$ \\[1mm]
           \mbox{for any vertical vector field $Y$ on~$P$}
          \end{center}}
          \big\}~.
\end{equation}

\section{Locally hamiltonian vector fields}

The results of the previous section, in particular equations~%
(\ref{eq:TSCPHS3}) and~(\ref{eq:TSCPHS4}), provide strong moti\-vation
for studying projectable vector fields $\, X = (X_P^{},X_M^{}) \,$ on
(ordinary) multiphase space which are locally hamiltonian (that is,
such that $X_P^{}$ is locally hamiltonian), since they imply that each
of these provides a functional vector field $\mathslf{X}$ on covariant
phase space defined by a simple algebraic composition rule:
\begin{equation}
 \mathslf{X}\;\![\phi]~=~\delta_X^{} \phi
 \qquad \mbox{for $\, \phi \in \mathscr{S}$}~.
\end{equation}
As we shall see in the next section, this functional vector field is also
hamiltonian with respect to the natural symplectic form on covariant phase
space to be presented there.
But before doing so, we want to address the problem of classifying the
locally hamiltonian vector fields and, along the way, also the globally
hamiltonian and exact hamiltonian vector fields on multiphase space.
For the case of extended multiphase space, using the forms $\omega$
and~$\theta$, this classification has been available in the literature
for some time~\cite{FPR1,FPR2}, even for the general case of multivector
fields. \linebreak
But what is relevant here is the corresponding result for the case
of ordinary multiphase space, using the forms $\omega_{\mathcal{H}}^{}$
and $\theta_{\mathcal{H}}^{}$, for a given hamiltonian $\mathcal{H}$.
As we shall see, there is one basic \linebreak phenomenon common to
both cases, namely that when $n\!>\!1$, the condition of a vector
field to be locally hamiltonian imposes severe restrictions on the
momentum dependence of its components, forcing them to be at most
affine (linear + constant): this appears to be a characteristic
feature distinguishing mechanics ($n\!=\!1$) from field theory
($n\!>\!1$) and has been noticed early by various authors (see,
e.g., \cite{Gaw,Kij}) and repeatedly rediscovered later on
(see, e.g., \cite{FHR,FPR1,FPR2,Ka1,Ka2}).
However, despite all similarities, some of the details depend on which
of the two types of multiphase space we are working with, so it seems
worthwhile to give a full statement for both cases, for the sake of
comparison.

We begin with an explicit computation in local Darboux coordinates
$(x^\mu,q^i,p_i^\mu)$ for~$\vec{J}^{\circledast} E$ and $(x^\mu,q^i,p_i^\mu,p)$
for $J^{\circledstar} E$ induced by the choice of local coordinates $x^\mu$
for~$M$, $q^i$ for the typical fiber $Q$ of~$E$ and a local
trivialization of~$E$ over~$M$, supposing that we are given
a hamiltonian $\, \mathcal{H}: \vec{J}^{\circledast} E \longrightarrow
J^{\circledstar} E \,$ which we represent in the form $\, \mathcal{H}
= - \, H \, d^{\,n} x$, as usual.
Starting out from equations~(\ref{eq:MCANF1})-(\ref{eq:MSPF2b}),
we distinguish two cases:

\subsection*{Extended multiphase space}

Using equations~(\ref{eq:MCANF1})-(\ref{eq:MSPF01}) and
writing an arbitrary vector field $X$ on $J^{\circledstar} E$ as
\begin{equation} \label{eq:VFEMPS} 
 X~=~X_{\vphantom{i}}^\mu \, \frac{\partial}{\partial x^\mu} \, + \,
     X^i \, \frac{\partial}{\partial q^i} \, + \,
     X_i^\mu \, \frac{\partial}{\partial p\>\!_i^\mu} \, + \,
     X_0^{} \, \frac{\partial}{\partial p}
\end{equation}
we first note that $X$ will be projectable to~$E$ if and only if the
coefficients $X_{\vphantom{i}}^\mu$ and $X^i$ do not depend on the energy
variable~$p$ nor on the multimomentum variables~$p_k^\kappa$ and will be
projectable to~$M$ if and only if the coefficients $X_{\vphantom{i}}^\mu$
do not depend on the energy variable~$p$ nor on the multimomentum
variables~$p_k^\kappa$ nor on the position variables $q^k$.
Next, we compute
\begin{equation} \label{eq:CONTR1E} 
 \begin{array}{rcl}
  i_X^{} \omega \!\!
  &=&\!\! X_{}^\nu \; dq^i \,\smwedge\, dp\>\!_i^\mu \,\smwedge\,
          d^{\,n} x_{\mu\nu} \, - \,
          X_i^\mu \; dq^i \,\smwedge\, d^{\,n} x_\mu \, + \,
          X_{}^i \; dp\>\!_i^\mu \,\smwedge\, d^{\,n} x_\mu \\[2mm]
  & &\!\! \mbox{} + \,
          X_{}^\mu \; dp \;\smwedge\, d^{\,n} x_\mu \, - \,
          X_0^{} \; d^{\,n} x~,
 \end{array}
\end{equation}
and
\begin{equation} \label{eq:CONTR2E} 
 i_X^{} \theta~
 =~(p\>\!_i^\mu X_{}^i \, + \, p \, X_{}^\mu) \; d^{\,n} x_\mu \, - \,
   p\>\!_i^\mu X_{}^\nu \; dq^i \,\smwedge\, d^{\,n} x_{\mu\nu}~.
\end{equation}
Then we have
\begin{thm} \label{thm:HVFEMPS}~
 Any locally hamiltonian vector field on~$J^{\circledstar} E$ is
 projectable to~$E$.
 Moreover, given any vector field $X$ on~$J^{\circledstar} E$
 which is projectable to~$E$, $X$ is locally hamiltonian if and
 only if
 \begin{enumerate}
  \item If $N\!>\!1$, $X$ is also projectable to~$M$, i.e., the
        coefficients $X_{\vphantom{i}}^\mu$ do not depend on the position
        variables~$q^k$.
  \item The coefficients $X_i^\mu$ and $X_0^{}$ can be expressed in
        terms of the previous ones and of new coefficients $X_-^\mu$
        which, once again, do not depend on the energy variable~$p$
        nor on the multimomentum variables $p_k^\kappa$, according to
        \begin{equation} \label{eq:LHVFEMPS1}
         X_i^\mu~
         =~- \, p \; \frac{\partial X_{}^\mu}{\partial q^i} \, - \,
           p\>\!_j^\mu \; \frac{\partial X^j}{\partial q^i} \, + \,
           p\>\!_i^\nu \; \frac{\partial X_{}^\mu}{\partial x^\nu} \, - \,
           p\>\!_i^\mu \; \frac{\partial X_{}^\nu}{\partial x^\nu} \, + \,
           \frac{\partial X_-^\mu}{\partial q^i}~,
        \end{equation}
        \begin{equation} \label{eq:LHVFEMPS2}
         X_0^{}~
         =~- \, p \; \frac{\partial X_{}^\mu}{\partial x_{}^\mu} \, - \,
             p\>\!_i^\mu \; \frac{\partial X_{}^i}{\partial x^\mu} \, + \,
             \frac{\partial X_-^\mu}{\partial x^\mu}~.
        \end{equation}
 \end{enumerate}
 Finally, $X$ is exact hamiltonian if and only if, in addition, the
 coefficients $X_-^\mu$ vanish.
\end{thm}

\begin{proof}
 The proof is carried out by ``brute force'' computation and a term by
 term analysis of the coefficents that appear in the Lie derivative of~%
 $\omega$ along~$X$: it will be omitted since an explicit proof of a
 more general statement along these lines can be found in~\cite{FPR2}.
\qed
\end{proof}

\subsection*{Ordinary multiphase space}

Using equations~(\ref{eq:MCANF2})-(\ref{eq:MSPF2b}) and
writing an arbitrary vector field~$X$ on~$\vec{J}^{\circledast} E$ as
\begin{equation} \label{eq:VFOMPS}
 X ~=~ X_{\vphantom{i}}^\mu \, \frac{\partial}{\partial x^\mu} \, + \,
       X^i \, \frac{\partial}{\partial q^i} \, + \,
       X_i^\mu \, \frac{\partial}{\partial p\>\!_i^\mu}~,
\end{equation}
we first note that $X$ will be projectable to~$E$ if and only if
the coefficients $X_{\vphantom{i}}^\mu$ and $X^i$ do not depend on the
multimomentum variables~$p_k^\kappa$ and will be projectable to~$M$
if and only if the coefficients $X_{\vphantom{i}}^\mu$ do not depend
on the multimomentum variables~$p_k^\kappa$ nor on the position
variables~$q^k$.
Next, we compute
\begin{equation} \label{eq:CONTR1O} 
 \begin{array}{rcl}
  i_X^{} \omega_{\mathcal{H}}^{} \!\!
  &=&\!\! {\displaystyle X^\nu \;
          dq^i \,\smwedge\, dp\>\!_i^\mu \,\smwedge\, d^{\,n} x_{\mu\nu} \; - \,
          \left( X_i^\mu \, + \, X_{\vphantom{i}}^\mu \;
                 \frac{\partial H}{\partial q^i} \right) \,
          dq^i \,\smwedge\, d^{\,n} x_\mu} \\[5mm]
  & &\!\! {\displaystyle \mbox{} + \, \left(
          X^i \delta_\mu^\nu \; - \;
          X^\nu \; \frac{\partial H}{\partial p\>\!_i^\mu} \right) \,
          dp\>\!_i^\mu \,\smwedge\, d^{\,n} x_\nu \; + \,
          \left( X^i \; \frac{\partial H}{\partial q^i} \; + \;
                 X_i^\mu \; \frac{\partial H}{\partial p\>\!_i^\mu} \right) \,
          d^{\,n} x}~,
 \end{array}
\end{equation}
and
\begin{equation} \label{eq:CONTR2O} 
 i_X^{} \theta_{\mathcal{H}}^{}~
 =~(p\>\!_i^\mu X^i \, - \, H X_{\vphantom{i}}^\mu) \; d^{\,n} x_\mu \, - \,
   p\>\!_i^\mu X^\nu \; dq^i \,\smwedge\, d^{\,n} x_{\mu\nu}~.
\end{equation}
Then we have
\begin{thm} \label{thm:HVFOMPS}~
 Any locally hamiltonian vector field on~$\vec{J}^{\circledast} E$ is
 projectable to~$E$, except possibly when $N\!=\!1$ and $n\!=\!2$,
 and any exact hamiltonian vector field on~$\vec{J}^{\circledast} E$
 is projectable to~$E$.
 Moreover, given any vector field $X$ on~$\vec{J}^{\circledast} E$
 which is projectable to~$E$, $X$ is locally hamiltonian if and
 only if
 \begin{enumerate}
  \item If $N\!>\!1$, $X$ is also projectable to~$M$, i.e., the
        coefficients $X_{\vphantom{i}}^\mu$ do not depend on the position
        variables~$q^k$.
  \item The coefficients $X_i^\mu$ can be expressed in terms of
        the previous ones and of new coefficients $X_-^\mu$ which,
        once again, do not depend on the multimomentum variables
        $p_k^\kappa$, according to\,%
        \footnote{Note that the first term in equation~(\ref{eq:LHVFOMPS1})
        is absent as soon as $N\!>\!1$.}
        \begin{eqnarray} \label{eq:LHVFOMPS1}
         X_i^\mu \!\!
         &=&\!\! H \; \frac{\partial X^\mu}{\partial q^i} \, - \,
                 p\>\!_j^\mu \; \frac{\partial X^j}{\partial q^i} \, + \,
                 p\>\!_i^\nu \; \frac{\partial X^\mu}{\partial x^\nu} \, - \,
                 p\>\!_i^\mu \; \frac{\partial X^\nu}{\partial x^\nu} \, + \,
                 \frac{\partial X_-^\mu}{\partial q^i}~.
        \end{eqnarray}
  \item The coefficients $X_{\vphantom{i}}^\mu$, $X^i$ and $X_-^\mu$
        satisfy the compatibility condition
        \begin{equation} \label{eq:LHVFOMPS2}
         \frac{\partial H}{\partial x^\mu} \, X^\mu \, + \,
         \frac{\partial H}{\partial q^i} \, X^i \, + \,
         \frac{\partial H}{\partial p\>\!_i^\mu} \, X_i^\mu~
         =~- \, H \; \frac{\partial X^\mu}{\partial x^\mu} \, + \,
                p\>\!_i^\mu \; \frac{\partial X^i}{\partial x^\mu} \, - \,
                \frac{\partial X_-^\mu}{\partial x^\mu}~.
        \end{equation}
 \end{enumerate}
 Finally, $X$ is exact hamiltonian if and only if, in addition,
 the coefficients $X_-^\mu$ vanish.
\end{thm}

\noindent
Maybe somewhat surprisingly, the case where $N\!=\!1$ and $n\!=\!2$,
concerning the theory of a single real scalar field in two space-time
dimensions, is somewhat exceptional in that it allows for locally
hamiltonian vector fields which fail to be projectable and are not
covered by the above classification theorem; we shall address this
question in Remark~\ref{ob:EXC} below.

\begin{proof}
 As in the case of Theorem~1, the proof is carried out by ``brute force''
 computation.
 First, we apply the exterior derivative to equation~(\ref{eq:CONTR1O})
 and collect terms to get
 \vspace{2mm}
 \begin{eqnarray*}
  L_X^{} \, \omega_{\mathcal{H}}^{} \!\!
  &=&\!\!  \bigg( \frac{\partial}{\partial x^\mu}
           \left( X_i^\mu \, + \,
                  X^\mu \; \frac{\partial H}{\partial q^i} \right) \; + \;
           \frac{\partial}{\partial q^i} \,
           \bigg( X^j \; \frac{\partial H}{\partial q^j} \; + \;
                  X_j^\nu \; \frac{\partial H}{\partial p\>\!_j^\nu} \bigg)
           \bigg) \; dq^i \,\smwedge\, d^{\,n} x                         \\[2mm]
  & & - \; \bigg( \frac{\partial}{\partial x^\nu}
           \left( \delta_\mu^\nu \, X^i \; - \;
                  X^\nu \; \frac{\partial H}{\partial p\>\!_i^\mu}
                  \right) \; - \;
           \frac{\partial}{\partial p\>\!_i^\mu} \,
           \bigg( X^j \; \frac{\partial H}{\partial q^j} \; + \;
                  X_j^\nu \; \frac{\partial H}{\partial p\>\!_j^\nu} \bigg)
           \bigg) \; dp\>\!_i^\mu \,\smwedge\, d^{\,n} x                 \\[2mm]
  & & + \; \bigg( \delta_l^k \, \bigg( \delta_\kappa^\sigma \,
           \frac{\partial X^\tau}{\partial x^\tau} \, - \,
           \frac{\partial X^\sigma}{\partial x^\kappa} \bigg)            \\[2mm]
  & & \qquad + \; \frac{\partial}{\partial p_k^\kappa}
           \left( X_l^\sigma \, + \,
                  X^\sigma \; \frac{\partial H}{\partial q^l} \right) \; + \;
           \frac{\partial}{\partial q^l}
           \left( \delta_\kappa^\sigma \, X^k \; - \;
                  X^\sigma \; \frac{\partial H}{\partial p_k^\kappa} \right)
           \bigg) \; dq^l \,\smwedge\,  dp\>\!_k^\kappa \,\smwedge\,
           d^{\,n} x_{\sigma}^{}                                         \\[2mm]
  & & - \; \frac{\partial}{\partial q^j}
           \left( X_i^\mu \, + \,
                  X^\mu \; \frac{\partial H}{\partial q^i} \right) \,
           dq^j \,\smwedge\, dq^i \,\smwedge\, d^{\,n} x_\mu^{}          \\[2mm]
  & & + \; \frac{\partial}{\partial p_k^\kappa}
           \left( \delta_\lambda^\sigma \, X^l \; - \; X^\sigma \;
                  \frac{\partial H}{\partial p\>\!_l^\lambda} \right) \,
           dp_k^\kappa \,\smwedge\, dp\>\!_l^\lambda \,\smwedge\,
           d^{\,n} x_\sigma^{}                                           \\[2mm]
  & & + \; \frac{\partial X_{}^\sigma}{\partial q^k} \;
           dq^k \,\smwedge\,  dq^l \,\smwedge\, dp\>\!_l^\lambda \,\smwedge\,
           d^{\,n} x_{\lambda\sigma}^{}                     \hspace{1cm} \\[2mm]
  & & - \; \frac{\partial X^\sigma}{\partial p_k^\kappa} \;
           dq^l \,\smwedge\, dp_k^\kappa \,\smwedge\,
           dp\>\!_l^\lambda \,\smwedge\, d^{\,n} x_{\lambda\sigma}^{}~.
 \rule[-6mm]{0mm}{8mm}
 \end{eqnarray*}
 Numbering the terms in this equation from $1$ to $7$, we begin by analyzing
 terms no.\ 7, 5 and 6. \linebreak
 Obviously, when $X$ is projectable to~$E$ and the coefficients $X^\mu$
 satisfy the condition stated in item 1.\ of the theorem, these terms
 vanish identically (term no.\ 6 is absent when $N\!=\!1$), so what
 we need to analyze is the converse statement.
 \begin{itemize}
  \item Term No.\ 7: For any choice of indices $\, i,j,m,\mu,\nu \,$
        and mutually different indices $\, \rho_1,\ldots,\rho_{n-2}$,
        contracting $L_X^{} \omega_{\mathcal{H}}^{}$ with the multivector
        field $\, \partial_m^{} \,\smwedge\, \partial_\mu^{\,i}
        \,\smwedge\, \partial_{\nu\vphantom{j}}^{\,j} \,\smwedge\,
        \partial_{\rho_1}^{} \,\smwedge\, \ldots \,\smwedge\,
        \partial_{\rho_{n-2}}^{}$ \linebreak gives the relation
        \begin{equation} \label{eq:LHVFOMPS3}
         \delta_m^j \; \frac{\partial X_{}^\sigma}{\partial p\>\!_i^\mu} \,
         \epsilon_{\nu \sigma \rho_1 \ldots \rho_{n-2}}^{}~
         =~\delta_m^i \; \frac{\partial X_{}^\sigma}{\partial p\>\!_j^\nu} \,
           \epsilon_{\mu \sigma \rho_1 \ldots \rho_{n-2}}^{}~.
        \end{equation}
        Now if to begin with, we fix only the indices $i$ and $\mu$,
        together with some other index $\rho$, \linebreak we can always
        choose the remaining free indices in this equation to be such
        that $m\!=\!j$ and $\, \nu,\rho_1,\ldots,\rho_{n-2} \,$ are all
        mutually different and $\neq \rho \,$: this reduces the lhs to
        the expression $\, \pm \, \partial X^\rho / \partial p\>\!_i^\mu$,
        while the rhs vanishes if we take $m\!\neq\!i$, which is possible
        as soon as $N\!>\!1$.
        Thus we conclude that
        \vspace{2mm}
        \[
         \frac{\partial X^\rho}{\partial p\>\!_i^\mu}~=~0~,
        \]
        except perhaps when $N\!=\!1$.
        But even when $N\!=\!1$, where equation~(\ref{eq:LHVFOMPS3})
        reduces to\,%
        \footnote{In the case $N\!=\!1$, the ``internal'' index on the
        position and multimomentum variables can only assume a single
        fixed value, say $i$, and so we could in principle just omit
        it, or else repeat it as often as we like.}
        \begin{equation} \label{eq:LHVFOMPS4}
         \frac{\partial X_{}^\sigma}{\partial p\>\!_i^\mu} \,
         \epsilon_{\nu \sigma \rho_1 \ldots \rho_{n-2}}^{}~
         =~\frac{\partial X_{}^\sigma}{\partial p\>\!_i^\nu} \,
           \epsilon_{\mu \sigma \rho_1 \ldots \rho_{n-2}}^{}~,
        \end{equation}
        this conclusion remains valid as long as $n\!>\!2$.
        Indeed, if we choose $\, \nu,\rho_1,\ldots,\rho_{n-2} \,$
        as above (all mutually different and $\neq \rho$), then
        if $\rho\!\neq\!\mu$, the lhs reduces to the expression
        $\, \pm \, \partial X^\rho / \partial p\>\!_i^\mu$, as before,
        while the rhs vanishes since in that case, $\mu$ must appear
        among the indices $\, \rho_1,\ldots,\rho_{n-2}$, whereas if
        $\rho=\mu$, the equation assumes the form
        \[
         \frac{\partial X^\mu}{\partial p\>\!_i^\mu} \,
         \epsilon_{\nu \mu \rho_1 \ldots \rho_{n-2}}^{}~
         =~\frac{\partial X^\nu}{\partial p\>\!_i^\nu} \,
           \epsilon_{\mu \nu \rho_1 \ldots \rho_{n-2}}^{}
         \qquad \mbox{(no sum over $\mu$ and  $\nu$)}~,
        \]
        which implies
        \[
         \frac{\partial X^\mu}{\partial p\>\!_i^\mu}~
         = \, - \, \frac{\partial X^\nu}{\partial p\>\!_i^\nu}~
         =~\frac{\partial X^\kappa}{\partial p\>\!_i^\kappa}~
         = \, - \, \frac{\partial X^\mu}{\partial p\>\!_i^\mu}
         \qquad \mbox{(no sum over $\mu$, $\nu$ and $\kappa$)}~,
        \]
        for mutually different $\, \mu,\nu,\kappa$.
        On the other hand, when $N\!=\!1$ and $n\!=\!2$, these arguments
        fail, and the only conclusion that can be drawn from equation~%
        (\ref{eq:LHVFOMPS4}) is that the following divergence must vanish:%
        \footnotemark[8]
        \begin{equation} \label{eq:N1n2T71}
         \frac{\partial X^\mu}{\partial p\>\!_i^\mu}~=~0~.
        \end{equation}
  \item Term No.\ 5: For any choice of indices $\, i,j,\mu,\nu \,$
        and mutually different indices $\, \rho_1,\ldots,\rho_{n-1}$,
        contracting $L_X^{} \omega_{\mathcal{H}}^{}$ with the multivector field
        $\, \partial_\mu^{\,i} \,\smwedge\, \partial_{\nu\vphantom{j}}^{\,j}
        \,\smwedge\, \partial_{\rho_1}^{} \,\smwedge\, \ldots \,\smwedge\,
        \partial_{\rho_{n-1}}^{} \,$ gives the relation
        \[
         \frac{\partial}{\partial p\>\!_i^\mu}
         \left( \delta_\nu^\sigma \, X^j \; - \;
                X^\sigma \; \frac{\partial H}{\partial p\>\!_j^\nu} \right)
         \epsilon_{\sigma \rho_1 \ldots \rho_{n-1}}^{}~
         =~\frac{\partial}{\partial p\>\!_j^\nu}
           \left( \delta_\mu^\sigma \, X^i \; - \;
                  X^\sigma \; \frac{\partial H}{\partial p\>\!_i^\mu} \right)
           \epsilon_{\sigma \rho_1 \ldots \rho_{n-1}}^{}~,
        \]
        so that, for any choice of indices $\, i,j,\mu,\nu,\rho$, taking
        $\, \rho_1,\ldots,\rho_{n-1} \,$ to be $\neq \rho$ shows that
        \begin{equation} \label{eq:LHVFOMPS5}
         \delta_\nu^\rho \,
         \frac{\partial X^j}{\partial p\>\!_i^\mu} \; - \;
         \frac{\partial X^\rho}{\partial p\>\!_i^\mu} \;
         \frac{\partial H}{\partial p\>\!_j^\nu}~
         =~\delta_\mu^\rho \,
           \frac{\partial X^i}{\partial p\>\!_j^\nu} \; - \;
           \frac{\partial X^\rho}{\partial p\>\!_j^\nu} \;
           \frac{\partial H}{\partial p\>\!_i^\mu}~.
        \end{equation}
        When $N\!>\!1$ or $n\!>\!2$, we can use the result of the previous
        item to conclude that
        \[
         \delta_\nu^\rho \, \frac{\partial X^j}{\partial p\>\!_i^\mu}~
         =~\delta_\mu^\rho \, \frac{\partial X^i}{\partial p\>\!_j^\nu}~.
        \]
        Now if to begin with, we fix only the indices $i,j$ and $\mu$,
        we can always choose the other free indices $\nu$ and~$\rho$ in
        this equation to be equal and $\neq \mu$: this reduces the lhs
        to the expression $\, \pm \, \partial X^j / \partial p\>\!_i^\mu$,
        while the rhs vanishes.
        Thus we conclude that
        \[
         \frac{\partial X^j}{\partial p\>\!_i^\mu}~=~0~.
        \]
        On the other hand, when $N\!=\!1$ and $n\!=\!2$, equation~%
        (\ref{eq:LHVFOMPS5}) reduces to a simple statement of symmetry:%
        \footnotemark[8]
        \begin{equation} \label{eq:N1n2T51}
         \epsilon^{\mu\nu} \left( \delta_\nu^\rho \,
         \frac{\partial X^i}{\partial p\>\!_i^\mu} \; - \;
         \frac{\partial X^\rho}{\partial p\>\!_i^\mu} \;
         \frac{\partial H}{\partial p\>\!_i^\nu} \right) =~0~.
        \end{equation}
  \item Term No.\ 6: For any choice of indices $\, i,j,m,\nu \,$ and
        mutually different indices $\, \rho_1,\ldots,\rho_{n-2}$,
        contracting $L_X^{} \omega_{\mathcal{H}}^{}$ with the multivector
        field $\, \partial_i^{} \,\smwedge\, \partial_j^{} \,\smwedge\,
        \partial_{\nu\vphantom{j}}^m \,\smwedge\, \partial_{\rho_1}^{}
        \,\smwedge\, \ldots \,\smwedge\, \partial_{\rho_{n-2}}^{} \,$
        gives the \linebreak relation
        \[
         \delta_j^m \; \frac{\partial X^\sigma}{\partial q^i} \,
         \epsilon_{\nu \sigma \rho_1 \ldots \rho_{n-2}}^{}~
         =~\delta_i^m \; \frac{\partial X^\sigma}{\partial q^j} \,
         \epsilon_{\nu \sigma \rho_1 \ldots \rho_{n-2}}^{}~.
        \]
        As before, this implies
        \[
         \frac{\partial X^\rho}{\partial q^i}~=~0~,
        \]
        except when $N\!=\!1$: in this case, the whole term vanishes
        identically and no conclusion can be drawn.
 \end{itemize}
 To proceed further, we write down the equations obtained from the
 remaining terms:
 \begin{itemize}
  \item Term No.\ 1:
        \begin{equation} \label{eq:LHAMMVF11}
         \frac{\partial}{\partial q^i} \,
         \bigg( X^j \; \frac{\partial H}{\partial q^j} \; + \;
                X_j^\nu \; \frac{\partial H}{\partial p\>\!_j^\nu} \, \bigg)~
         = \, - \, \frac{\partial}{\partial x^\mu}
              \left( X_i^\mu \, + \,
                     X^\mu \; \frac{\partial H}{\partial q^i} \right) .
        \end{equation}
  \item Term No.\ 2:
        \begin{equation} \label{eq:LHAMMVF12}
         \frac{\partial}{\partial p\>\!_i^\mu} \,
         \bigg( X^j \; \frac{\partial H}{\partial q^j} \; + \;
                X_j^\nu \; \frac{\partial H}{\partial p\>\!_j^\nu} \bigg)~
         =~\frac{\partial}{\partial x^\nu}
           \left( \delta_\mu^\nu \, X^i \; - \;
                  X^\nu \; \frac{\partial H}{\partial p\>\!_i^\mu} \right) .
        \end{equation}
  \item Term No.\ 3:
        \begin{equation} \label{eq:LHAMMVF13}
         \begin{array}{rcl}
          {\displaystyle
           \frac{\partial}{\partial p\>\!_j^\nu}
           \left( X_i^\mu \, + \,
                  X^\mu \; \frac{\partial H}{\partial q^i} \right) \!\!}
          &=&\!\! {\displaystyle \mbox{} - \,
                   \frac{\partial}{\partial q^i}
                   \left( \delta_\nu^\mu \, X^j \; - \;
                          X^\mu \; \frac{\partial H}{\partial p\>\!_j^\nu}
                          \right)} \\[6mm]
          & &\!\! {\displaystyle \mbox{} + \,
                   \delta_i^j \;
                   \frac{\partial X^\mu}{\partial x^\nu} \, - \,
                   \delta_\nu^\mu \, \delta_i^j \;
                   \frac{\partial X^\kappa}{\partial x^\kappa}~.}
         \end{array}
        \end{equation}
  \item Term No.\ 4:
        \begin{equation} \label{eq:LHAMMVF14}
         \frac{\partial X_{i\vphantom{j}}^\mu}{\partial q^j}~
         =~\frac{\partial X_j^\mu}{\partial q^i}~.
        \end{equation}
 \end{itemize}
 Assuming that $X$ is projectable to~$E$, we observe that equation~%
 (\ref{eq:LHAMMVF13}) can be integrated directly to conclude that
 \[
  X_i^\mu~=~H \; \frac{\partial X^\mu}{\partial q^i} \, - \,
            p\>\!_j^\mu \; \frac{\partial X^j}{\partial q^i} \, + \,
            p\>\!_i^\nu \; \frac{\partial X^\mu}{\partial x^\nu} \, - \,
            p\>\!_i^\mu \; \frac{\partial X^\kappa}
                                {\partial x^\kappa} \, + \, Y_i^\mu~,
 \]
 where the $Y_i^\mu$ are local functions on~$E$ which, once again,
 are independent of the multimomentum variables $p_k^\kappa$.
 Substituting this relation into equation~(\ref{eq:LHAMMVF14}),
 we get
 \[
  \frac{\partial Y_{i\vphantom{j}}^\mu}{\partial q^j}~
  =~\frac{\partial Y_j^\mu}{\partial q^i}~,
 \]
 which can be solved by setting
 \[
  Y_i^\mu~=~\frac{\partial Y_-^\mu}{\partial q^i}~,
 \]
 where the $Y_-^\mu$ are local functions on~$E$ which, as before,
 are independent of the multimomentum variables $p_k^\kappa$.
 Finally, substituting this expression into equations~(\ref{eq:LHAMMVF11})
 and (\ref{eq:LHAMMVF12}), we get
 \[
  \frac{\partial}{\partial q^i}
  \left( H \; \frac{\partial X^\nu}{\partial x^\nu} \, - \,
         p\>\!_j^\nu \; \frac{\partial X^j}{\partial x^\nu} \, + \,
         \frac{\partial Y_-^\nu}{\partial x^\nu} \, + \,
         \frac{\partial H}{\partial x^\nu} \, X^\nu \, + \,
         \frac{\partial H}{\partial q^j} \, X^j \, + \,
         \frac{\partial H}{\partial p\>\!_j^\nu} \, X_j^\nu \right) \, =~0~,
 \]
 and
 \[
  \frac{\partial}{\partial p\>\!_i^\mu}
  \left( H \; \frac{\partial X^\nu}{\partial x^\nu} \, - \,
         p\>\!_j^\nu \; \frac{\partial X^j}{\partial x^\nu} \, + \,
         \frac{\partial Y_-^\nu}{\partial x^\nu} \, + \,
         \frac{\partial H}{\partial x^\nu} \, X^\nu \, + \,
         \frac{\partial H}{\partial q^j} \, X^j \, + \,
         \frac{\partial H}{\partial p\>\!_j^\nu} \, X_j^\nu \right) \, =~0~,
 \]
 showing that
 \[
  H \; \frac{\partial X^\nu}{\partial x^\nu} \, - \,
  p\>\!_j^\nu \; \frac{\partial X^j}{\partial x^\nu} \, + \,
  \frac{\partial Y_-^\nu}{\partial x^\nu} \, + \,
  \frac{\partial H}{\partial x^\nu} \, X^\nu \, + \,
  \frac{\partial H}{\partial q^j} \, X^j \, + \,
  \frac{\partial H}{\partial p\>\!_j^\nu} \, X_j^\nu~=~Y_-^{}~,
 \]
 where $Y_-^{}$ is a local function on~$M$ which is independent of
 the position variables $q^k$ and multimomentum variables $p_k^\kappa$.
 Writing $Y_-^{}$ as a divergence,
 \[
  Y_-^{}~=~\frac{\partial Y_-^{\prime\,\mu}}{\partial x^\mu}~,
 \]
 and putting $\, X^\mu = Y_-^\mu - Y_-^{\prime\,\mu}$, we arrive at
 equations~(\ref{eq:LHVFOMPS1}) and~(\ref{eq:LHVFOMPS2}).
 \\[2mm]
 All that remains to be shown now is the final statement concerning exact
 hamiltonian vector fields. To this end, we apply the exterior derivative
 to equation~(\ref{eq:CONTR2O}) and subtract the expression in equation~%
 (\ref{eq:CONTR1O}); then collecting terms, we get
 \vspace{2mm}
 \begin{eqnarray*}
  L_X^{} \, \theta_{\mathcal{H}}^{} \!\!
  &=&\!\!  d \big( i_X^{} \, \theta_{\mathcal{H}}^{} \big) \, - \,
           i_X \, \omega_{\mathcal{H}}^{} \\[4mm]
  &=&\!\!  \bigg( - \; \frac{\partial X^\mu}{\partial x^\mu} \, H \; + \,
                  \frac{\partial X^i}{\partial x^\mu} \, p\>\!_i^\mu \, -  \,
                  \left( \frac{\partial H}{\partial x^\mu} \, X^\mu \, + \,
                         \frac{\partial H}{\partial q^i} \, X^i \, + \,
                         \frac{\partial H}{\partial p\>\!_i^\mu} \, X_i^\mu
                         \right) \bigg) \; d^{\,n} x \\[2mm]
  & & + \; \bigg( \frac{\partial X^j}{\partial q^i} \, p\>\!_j^\mu \, - \,
                  \frac{\partial X^\mu}{\partial x^\nu} \, p\>\!_i^\nu \, + \,
                  \frac{\partial X^\nu}{\partial x^\nu} \, p\>\!_i^\mu \, - \,
                  \frac{\partial X^\mu}{\partial q^i} \, H \; + \,
                  X_i^\mu \bigg) \;
           dq^i \,\smwedge\, d^{\,n} x_{\mu} \\[2mm]
  & & + \; \bigg( \frac{\partial X^i}{\partial p\>\!_j^\nu} \,
                  p\>\!^\mu_i \, - \,
                  \frac{\partial X^\mu}{\partial p\>\!_j^\nu} \, H \bigg) \;
           dp\>\!_j^\nu \,\smwedge\, d^{\,n} x_{\mu} \\[2mm]
  & & + \; \frac{\partial X^\nu}{\partial q^j} \, p\>\!_i^\mu \;
           dq^i \,\smwedge\, dq^j \,\smwedge\, d^{\,n} x_{\mu\nu} \\[2mm]
  & & + \; \frac{\partial X^\nu}{\partial p\>\!_k^\kappa} \, p\>\!_i^\mu \;
           dq^i \,\smwedge\, dp\>\!^\kappa_k \,\smwedge\,
           d^{\,n} x_{\mu\nu}~.
 \end{eqnarray*}
 Numbering the terms in this equation from $1$ to $5$, we see that the
 conditions imposed by the fact that $X$ should be exact hamiltonian
 are the following:
 \begin{itemize}
  \item Term No.\ $5$: This term vanishes if and only if the coefficients
        $X^\mu$ do not depend on the variables $p\>\!_k^\kappa$.
  \item Term No.\ $3$: Due to the previous condition, this term vanishes
        if and only the coefficients $X^i$ do not depend on the variables
        $p\>\!_k^\kappa$.
  \item Term No.\ $4$: This term vanishes if and only if the coefficients
        $X^\mu$ do not depend on the variables $q^k$, except when
        $N\!=\!1$: in this case the whole term vanishes identically
        and no conclusion can be drawn.
  \item Term No.\ $2$: This term vanishes if and only if the coefficients
        $X_i^\mu$ are defined in terms of the coefficients $X^\mu$ and $X^i$
        according to equation~(\ref{eq:LHVFOMPS1}), with $\, X_-^\mu = 0$.
  \item Term No.\ $1$: This term vanishes if and only if equation~%
        (\ref{eq:LHVFOMPS2}) is required to hold, with $\, X_-^\mu = 0$.
 \end{itemize}
\qed
\end{proof}

\begin{rmk} \label{ob:EXC}~
 The classification of locally hamiltonian vector fields provided
 by Theorem~\ref{thm:HVFOMPS} is not quite complete since it does
 not cover non-projectable locally hamiltonian vector fields.
 \linebreak
 This may not be a reason for great concern since such vector
 fields are pathological in the sense that their flows do not
 respect any of the bundle structures involved and, perhaps
 more importantly, since such vector fields can only exist
 in one very special and exceptional case, namely when
 $N\!=\!1$ and $n\!=\!2$.
 Still, it is somewhat annoying that they do not seem to
 admit any reasonable classification.
 To give an idea of what is involved, consider first the
 more general case of one degree of freedom in any space-time
 dimension ($N\!=\!1$, $n\!>\!1$), where it is common practice
 to omit the ``internal'' index~$i$ on the variables $q^i$
 and~$p\>\!_i^\mu$; it is then appropriate to redefine the
 components of~$X$ in equation~(\ref{eq:VFOMPS}), say by
 writing
 \begin{equation}
  X~=~X^\mu \, \frac{\partial}{\partial x^\mu} \, + \,
      \tilde{X} \, \frac{\partial}{\partial q} \, + \,
      \tilde{X}^\mu \, \frac{\partial}{\partial p\>\!^\mu}~.
 \end{equation}
 Then when $n\!=\!2$, the arguments presented in the proof above
 do not allow to conclude that the coefficients $X^\mu$ and $\tilde{X}$
 are independent of the multimomentum variables $p^\kappa$, but only that
 \begin{equation}  \label{eq:N1n2T72}
  \frac{\partial X^\mu}{\partial p\>\!^\mu}~=~0~,
 \end{equation}
 as stated in equation~(\ref{eq:N1n2T71}), and
 \begin{equation} \label{eq:N1n2T52}
  \frac{\partial \tilde{X}}{\partial p\>\!^\mu} \, - \,
  \frac{\partial X^\nu}{\partial p\>\!^\mu} \;
  \frac{\partial H}{\partial p\>\!^\nu}~=~0~,
 \end{equation}
 which, taking into account equation~(\ref{eq:N1n2T71}) and using
 the relation $\, \pm \epsilon_{\mu\nu}^{} \epsilon^{\kappa\lambda} =
 \delta_\mu^\kappa \delta_\nu^\lambda - \delta_\nu^\kappa \delta_\mu^\lambda \,$
 (with any fixed sign convention for $\pm$), is easily shown to be
 equivalent to equation~(\ref{eq:N1n2T51}).
 Introducing a new function~$F$ defined as
 \begin{equation} \label{eq:N1n2FUNC}
  F~=~\tilde{X} \, - \, X^\mu \, \frac{\partial H}{\partial p\>\!^\mu}~,
 \end{equation}
 together with the Hessian matrix $H_{\mu\nu}^{}$ of the hamiltonian
 function $H$,
 \begin{equation} \label{eq:N1n2HESS}
  H_{\mu\nu}^{}~
  =~\frac{\partial^2 H}{\partial p\>\!^\mu \, \partial p\>\!^\nu}~,
 \end{equation}
 and its inverse $H^{\mu\nu}$, it is possible to express all coefficients
 of~$X$ in terms of $F$ and $H$ and their partial derivatives up to first
 order (for~$F$) or second order (for~$H$):
 \begin{equation} \label{eq:N1n2LH1}
  X^\mu~= \; - \, H^{\mu\nu} \, \frac{\partial F}{\partial p\>\!^\nu}~,
 \end{equation}
 \begin{equation} \label{eq:N1n2LH2}
  \tilde{X}~=~F \, - \, H^{\mu\nu} \, \frac{\partial H}{\partial p\>\!^\mu} \,
                                     \frac{\partial F}{\partial p\>\!^\nu}~,
 \end{equation}
 \vspace{1ex}
 \begin{equation} \label{eq:N1n2LH3}
  \begin{array}{rcl}
   \tilde{X}^\mu \!\!
   &=&\!\! {\displaystyle H^{\mu\nu}
            \left( \frac{\partial F}{\partial x^\nu} \, + \,
                   \frac{\partial H}{\partial p^\nu} \,
                   \frac{\partial F}{\partial q} \, - \,
                   \frac{\partial^{\,2} H}{\partial q \, \partial p^\nu} \, F
           \right)} \\[2ex]
   & &\!\! {\displaystyle \mbox{} - \, H^{\mu\kappa}
           \left(  \frac{\partial^{\,2} H}{\partial x^\kappa \, \partial p^\lambda}
           \, - \, \frac{\partial^{\,2} H}{\partial x^\lambda \, \partial p^\kappa}
           \, + \, \frac{\partial H}{\partial p^\kappa} \,
                   \frac{\partial^{\,2} H}{\partial q \, \partial p^\lambda}
           \, - \, \frac{\partial H}{\partial p^\lambda} \,
                   \frac{\partial^{\,2} H}{\partial q \, \partial p^\kappa}
           \right) H^{\lambda\nu} \, \frac{\partial F}{\partial p^\nu}~.}
  \end{array}
 \end{equation}
 However, finding the general solution of the entire system seems to be an
 exceedingly difficult task, except if one makes some simplifying assumptions
 on the hamiltonian~$H$. One obvious choice would be to take
 \begin{equation} \label{eq:N1n2HAM}
  H~=~{\textstyle \frac{1}{2}} \, g_{\mu\nu}^{}(x) \, p^\mu p^\nu \, + \,
      A_\mu^{}(x) \, p^\mu \, + \, V(x,q)~,
 \end{equation}
 where $g$ represents a Lorentz metric, $A$ is a gauge potential and
 $V$ is some scalar potential, but even in this situation we have not
 come to a definite conclusion.
 The only case in which a complete solution has been found is in the
 absence of external fields, i.e., when the metric tensor~$g$ and the
 scalar potential~$V$ are both independent of~$x$ whereas the gauge
 potential~$A$  vanishes, so $M$ is two-dimensional Minkowski space
 $\mathbb{R}^2$ and $g$ is the standard Minkowski metric~$\eta$;
 see~\cite[Appendix]{Kij}.
\end{rmk}

\section{Covariant phase space}

\subsection{Symplectic structure on covariant phase space}

One of the most important properties of the covariant phase space
$\mathscr{S}$ introduced above (see equations~(\ref{eq:COVPHS1})-%
(\ref{eq:COVPHS4})) is that it carries a naturally defined symplectic
structure~\cite{CW,Cr,Zu} which can in fact be derived immediately
from the multisymplectic structure on multiphase space~\cite{FSR}.
Namely, generalizing the prescription of equation~(\ref{eq:LFUNCT1})
in the sense of using ordinary differential forms on multiphase space
to produce functional differential forms, rather than just functionals,
we can define functional canonical $1$-forms $\Theta_{K_\Sigma}^{}$ and
$2$-forms $\Omega_{K_\Sigma}^{}$ on~$\mathscr{C}$, where $\Sigma$ is
a hypersurface in~$M$ (typically, when a Lorentz metric is given, a
Cauchy surface) and $K_\Sigma$ runs through the compact submanifolds
of~$\Sigma$ which are the closure of their interior in~$\Sigma$ and
have smooth boundary $\partial K_\Sigma$, by setting
\begin{equation} \label{eq:FUNCTCF} 
 (\Theta_{K_\Sigma}^{})_\phi^{}(\delta_X^{} \phi)~
 =~\int_{K_\Sigma} \phi^* (i_X^{} \theta_{\mathcal{H}}^{\vphantom{j}})
\end{equation}
for~$\, \phi \in \mathscr{C} \,$ and $\, \delta_X^{} \phi \in
T_\phi^{} \mathscr{C} \,$ with $X$ vertical, and
\begin{equation} \label{eq:FUNCTSF} 
 (\Omega_{K_\Sigma}^{})_\phi^{}(\delta_{X_1}^{} \phi,\delta_{X_2}^{} \phi)~
 =~\int_{K_\Sigma} \phi^* (i_{X_2}^{} i_{X_1}^{}
                           \omega_{\mathcal{H}}^{\vphantom{j}})
\end{equation}
for $\, \phi \smin \mathscr{C} \,$ and $\, \delta_{X_1}^{} \phi,
\delta_{X_2}^{} \phi \in T_\phi \mathscr{C} \,$ with $X_1,X_2$ vertical.
(The same formulas continue to hold if we require the vector fields
$X,\,X_1,\,X_2$ to be only vertical on the image of~$\phi$.)
As observed, e.g., in~\cite{CW,Cr,Zu} and, in the present context,
in~\cite{FSR}, the restriction of the form $\Omega_{K_\Sigma}^{}$
to~$\mathscr{S}$ does not depend on the submanifold~$K_\Sigma$,
provided that appropriate boundary conditions are imposed:\,%
\footnote{For example, when we compare the integral in equation~%
(\ref{eq:FUNCTSF}) over two compact submanifolds $K_{1,\Sigma_1}$
and $K_{2,\Sigma_2}$ of hypersurfaces $\Sigma_1$ and~$\Sigma_2$
in~$M$ whose union $K_{1,\Sigma_1} \cup K_{2,\Sigma_2}$ is the
boundary $\partial K$ of a compact submanifold $K$ of~$M$, we
will get the same result. This also happens when $K_{1,\Sigma_1}
\cup K_{2,\Sigma_2}$ is just part of the boundary of a compact
submanifold $K$ of~$M$ but the remainder, $\partial K \setminus
(K_{1,\Sigma_1} \cup K_{2,\Sigma_2})$, has empty intersection
with the intersection of the base supports of $\delta_{X_1}^{}
\phi$ and~$\delta_{X_2}^{} \phi$.}
this happens because when $\, \phi \smin \mathscr{S} \,$ and
$\, \delta_{X_1}^{} \phi , \delta_{X_2}^{} \phi \in T_\phi \mathscr{S}$,
the expression under the integral in equation~(\ref{eq:FUNCTSF}) is
a closed form (called the ``symplectic current'' in~\cite{CW}), so
that according to Stokes' theorem, its integral over any compact
submanifold without boundary vanishes.
Thus, at least formally, covariant phase space becomes a symplectic
manifold~-- albeit an infinite-dimensional one; its symplectic form
will in what follows be simply denoted by~$\Omega$ and is explicitly
given by the formula
\begin{equation} \label{eq:SFCPHS} 
 \Omega_\phi^{}(\delta_{X_1}^{} \phi,\delta_{X_2}^{} \phi)~
 =~\int_\Sigma \phi^* (i_{X_2}^{} i_{X_1}^{}
                       \omega_{\mathcal{H}}^{\vphantom{j}})
\end{equation}
where $\Sigma$ is any Cauchy surface in~$M$ and where $\, \phi \smin
\mathscr{S} \,$ and $\, \delta_{X_1}^{} \phi,\delta_{X_2}^{} \phi \in
T_\phi \mathscr{S}$, with $X_1,X_2$ vertical (or possibly just vertical on
the image of~$\phi$) and such that $\; \mathrm{supp} \; \delta_{X_1}^{} \phi
\,\cap\, \mathrm{supp} \; \delta_{X_2}^{} \phi \,\cap\, \Sigma$ \linebreak
is compact.

\subsection{Functional hamiltonian vector fields and Poisson brackets}

The central result obtained in Ref.~\cite{FSR} can be summarized in the form
of two theorems which we state explicitly because they form the background
for the work reported here.
The basic object that appears there is the \emph{Jacobi operator}\/
$\mathscr{J}[\phi]$, obtained by linearizing the De\,Donder\,--\,Weyl
operator around a solution $\, \phi \in \mathscr{S} \,$ and whose
kernel is precisely the space $T_\phi \mathscr{S}$ of solutions of
the linearized equations of motion, and its causal Green function~%
$G_\phi$.
\begin{thm}~ \label{thm:FHVF1} 
 Given a functional $\mathslf{F}$ with temporally compact support
 on covariant phase space~$\mathscr{S}$, the functional hamiltonian
 vector field $\mathslf{X}_{\mathsmf{F}}^{}$ on\/~$\mathscr{S}$
 associated to~$\mathslf{F}$, as defined by the formula
 \begin{equation} \label{eq:FHVF1} 
  \Omega_\phi^{} \bigl( \mathslf{X}_{\mathsmf{F}}^{}\,[\phi],\delta\phi \bigr)~
  =~\mathslf{F}^{\>\prime}[\phi] \cdot \delta\phi
  \qquad \mbox{for $\, \phi \in \mathscr{S}$,
                   $\delta\phi \in T_\phi^{} \mathscr{S}$}~,
 \end{equation}
 is given by ``convolution'' of the variational derivative of~$\mathslf{F}$
 (see equation~(\ref{eq:FUNDER})) with the causal Green function of the
 corresponding Jacobi operator:
 \begin{equation} \label{eq:FHVF2} 
  \mathslf{X}_{\mathsmf{F}}^{\>\vphantom{j}i}\,[\phi](x)~
  =~\int_M d^{\,n} y~G_\phi^{ij}(x,y) \,
    \frac{\delta\mathslf{F}}{\delta\phi^j}[\phi] \, (y)
  \qquad \mbox{for $\, \phi \in \mathscr{S}$}~.
 \end{equation}
\end{thm}
Note that the condition that $\mathslf{F}\,$ should have temporally compact
support will guarantee that both sides of equation~(\ref{eq:FHVF1}) make
sense provided we interpret $T_\phi^{} \mathscr{S}$ as being the space of
solutions of the linearized equations of motion of spatially compact
support, i.e., we regard it as the subspace of the space $T_\phi^{}
\mathscr{C}$ given by equations~(\ref{eq:TSCPHS1})-(\ref{eq:TSCPHS4})
where the latter is defined according to equation~(\ref{eq:FTCTS3}).

With this statement at hand, it is easy to write down the Poisson bracket
of two functionals $\mathslf{F}\,$ and $\mathslf{G}\,$\/ on~$\mathscr{S}\,$:
in complete analogy with the formula $\; \{f,g\} = i_{X_g}^{} i_{X_f}^{}
\omega = - dg(X_f^{}) \,$ from mechanics, it can be defined by
\begin{equation} \label{eq:POISB5}
 \{ \mathslf{F}\, , \mathslf{G}\, \}[\phi]~
 = \; - \, \mathslf{G}^{\>\prime}[\phi]
   \bigl( \mathslf{X}_{\mathsmf{F}}^{}\,[\phi] \bigr)
 \qquad \mbox{for $\, \phi \in \mathscr{S}$}~,
\end{equation}
or more explicitly,
\begin{equation} \label{eq:POISB6}
 \{ \mathslf{F}\, , \mathslf{G}\, \}[\phi]~
 = \; - \, \int_M d^{\,n} x~
   \frac{\delta \mathslf{G}}{\delta\phi^{\,k}}[\phi](x) \;
   \mathslf{X}_{\mathsmf{F}}^{\>\vphantom{j}k}\,[\phi](x)
  \qquad \mbox{for $\, \phi \in \mathscr{S}$}~.
\end{equation}
Combining this expression with that given in Theorem~\ref{thm:FHVF1},
we arrive at the second main conclusion:
\begin{thm}~ \label{thm:FPBR1} 
 Given two functionals $\mathslf{F}$ and $\mathslf{G}$ with temporally
 compact support on covariant phase space\/~$\mathscr{S}$, their Poisson
 bracket\/ $\{\mathslf{F}\,,\mathslf{G}\,\}$, with respect to the symplectic
 form\/~$\Omega$ introduced above, is precisely their Peierls\,--\,De\,Witt
 bracket, given by
 \begin{equation} \label{eq:POISB7}
  \{\mathslf{F}\,,\mathslf{G}\,\}[\phi]~
  =~\int_M d^{\,n} x \int_M d^{\,n} y~
    \frac{\delta \mathslf{F}}{\delta\phi^{\,k}}[\phi](x) \;
    G_\phi^{kl}(x,y) \;
    \frac{\delta \mathslf{G}}{\delta\phi^{\,l}}[\phi](y)
  \qquad \mbox{for $\, \phi \in \mathscr{S}$}~.
 \end{equation}
\end{thm}
Note that in view of the regularity conditions imposed to arrive at these
results, the previous constructions do not apply directly to degenerate
systems such as gauge theories: these require a separate treatment.

\subsection{The main theorems}

In this subsection, we present the two main theorems of the present paper
which, for local functionals of the form given by equation~(\ref{eq:LFUNCT1}),
provide a simple \emph{algebraic}\/ construction of the functional hamiltonian
vector field associated with such a functional and, as a corollary, a simple
\emph{algebraic} formula for the Poisson bracket of two such functionals.
That such formulas should exist is not at all obvious, taking into
account that the corresponding formulas for general functionals, as
given in Theorem~\ref{thm:FHVF1} and Theorem~\ref{thm:FPBR1} above,
are essentially analytic: to apply them in concrete examples, one
needs to solve a system of (linear) partial differential equations
in order to calculate the corresponding causal Green function.
Surprisingly, for local functionals of the form given by
equation~(\ref{eq:LFUNCT1}), the ``convolution type'' integral
operator which has this Green function as its kernel collapses.
\begin{thm}~ \label{thm:FHVF2}
 Suppose we are given a fiber bundle\/~$E$ (the field configuration
 bundle) over an \linebreak $n$-dimensional globally hyperbolic
 space-time manifold\/~$M$ and a hamiltonian  $\, \mathcal{H}:
 \vec{J}^{\circledast} E \longrightarrow J^{\circledstar} E$, which is
 a section of extended multiphase space\/~$J^{\circledstar} E$ over
 ordinary multiphase space\/~$\vec{J}^{\circledast} E$, \linebreak
 together with a hamiltonian\/ $(n\!-\!1)$-form $f$ on\/~%
 $\vec{J}^{\circledast} E$ of spatially compact support such
 that the corresponding hamiltonian vector field $X_f^{}$
 on\/~$\vec{J}^{\circledast} E$, defined by the formula
 $\, i_{X_f}^{} \omega_{\mathcal{H}}^{} = df$, is projectable
 to~$M$.%
 \footnote{Recall that according to Theorem~\ref{thm:HVFOMPS},
 this is automatic if~$N>1$.}
 Then given any Cauchy surface\/~$\Sigma$ in\/~$M$ and writing
 $\mathcal{F}_{\Sigma,f}^{}$ for the local functional on\/~%
 $\mathscr{S}$ associated to~$\Sigma$ and~$f$, as in Definition~%
\ref{def:LFUNCT}, and $\mathslf{X}_{\mathcal{F}_{\Sigma,f}}^{}$
 for the functional hamiltonian vector field on\/~$\mathscr{S}$
 associated to~$\mathcal{F}_{\Sigma,f}^{}$, as in Theorem~%
 \ref{thm:FHVF1}, we have
 \begin{equation} \label{eq:FHVF3} 
  {\mathslf{X}}_{\mathcal{F}_{\Sigma,f}}^{}[\phi]~=~\delta_{X_f}^{} \phi
  \qquad \mbox{for $\, \phi \in \mathscr{S}$}~.
 \end{equation}
\end{thm}

\noindent
Note that this is a version of Noether's theorem, i.e., we are dealing
with a conservation law, for the lhs of equation~(\ref{eq:FHVF3}) appears
to depend on~$\Sigma$ and the theorem states that it is equal to the rhs,
which does not!

\begin{proof}
 Rather than analyzing the integral formula~(\ref{eq:FHVF2}), we shall
 show directly that the algebraic formula~(\ref{eq:FHVF3}) satisfies
 all required conditions.
 First of all, we note that $\delta_{X_f}^{} \phi$ has spatially compact
 support because $f$ does.
 Moroever, it is clear from the results of Section~4 that equation~%
 (\ref{eq:FHVF3}) does provide a functional vector field not only on~%
 $\mathscr{C}$ but also on~$\mathscr{S}$ since, according to equation~%
 (\ref{eq:TSCPHS3}), $X_f^{}$ being locally hamiltonian with respect to~%
 $\omega_{\mathcal{H}}^{}$ and projectable to~$M$ implies that $\, \delta_{X_f}^{}
 \phi \in T_\phi^{} \mathscr{S} \,$ when $\, \phi \in \mathscr{S} \,$.
 Therefore, all that needs to be verified is that the expression given
 in equation~(\ref{eq:FHVF3}) satisfies the condition~(\ref{eq:FHVF1}).
 To this end, let $X_{f,M}^{}$ denote the projection of~$X_f^{}$ to~$M$
 and, for any given $\, \phi \in \mathscr{S}$, apply Lemma~%
 \ref{lem:VFREPVAR} to construct some projectable vector field
 $\tilde{X}_f^\phi$ on~$\vec{J}^{\circledast} E$ which is $\phi$-related
 to $X_{f,M}^{}$; then the difference $\, X_f^{} - \tilde{X}_f^\phi \,$
 will be vertical on the image of~$\phi$, and according to equation~%
 (\ref{eq:VFREPVAR1}), equation~(\ref{eq:FHVF3}) becomes
 \[
  \mathslf{X}_{\mathcal{F}_{\Sigma,f}}^{}[\phi]~
  =~(X_f^{} - \tilde{X}_f^\phi)(\phi)~
  =~\delta_{(X_f^{} - \tilde{X}_f^\phi)}^{} \phi~.
 \]
 Inserting this expression into equation~(\ref{eq:SFCPHS}), we get,
 for any $\, \delta_X^{} \phi \in T_\phi^{} \mathscr{S} \,$ where~$X$
 is some vertical vector field on~$\vec{J}^{\circledast} E$ of spatially
 compact support
 \[
  \Omega_\phi \bigl( \mathslf{X}_{\mathcal{F}_{\Sigma,f}}^{}[\phi],
                    \delta_X^{} \phi \bigr)~
  =~\Omega_\phi \bigl( \delta_{(X_f^{} - \tilde{X}_f^\phi)}^{} \phi,
                      \delta_X^{} \phi \bigr)~
  =~\int_\Sigma \phi^* \bigl( i_{X\vphantom{\tilde{X}_f^\phi}}^{}
                              i_{(X_f^{} - \tilde{X}_f^\phi)}^{}
                              \omega_{\mathcal{H}}^{} \bigr)
 \]
 Now observe that the expression
 \[
  \phi^* \bigl( i_{X\vphantom{\tilde{X}_f^\phi}}^{}
                i_{\tilde{X}_f^\phi}^{} \omega_{\mathcal{H}}^{} \bigr)~
  = \, - \, \phi^* \bigl( i_{\tilde{X}_f^\phi}^{}
                          i_{X\vphantom{\tilde{X}_f^\phi}}^{}
                          \omega_{\mathcal{H}}^{} \bigr)~
  = \, - \, i_{X_{f,M}^{}}^{}
            \bigl( \phi^* (i_X^{} \omega_{\mathcal{H}}^{}) \bigr)
 \]
 vanishes due to the assumption that $\phi$ is a solution of the equations
 of motion.
 Therefore,
 \[
  \Omega_\phi(\mathslf{X}_{\mathcal{F}_{\Sigma,f}}^{}[\phi],\delta_X^{} \phi)~
  =~\int_\Sigma \phi^* \bigl( i_X^{} i_{X_f}^{} \omega_{\mathcal{H}}^{} \bigr)~,
 \]
 which, according to Proposition~\ref{prop:FFUNDER}, is equal to
 \[
  \mathcal{F}_{\Sigma,f}^{\>\prime}[\phi] \cdot \delta_X^{} \phi~
  =~\int_\Sigma \phi^* \bigl( i_X^{} df \bigr)~
  =~\int_\Sigma \phi^* \bigl( i_X^{} i_{X_f}^{} \omega_{\mathcal{H}}^{} \bigr)~.
 \]
\qed
\end{proof}

An immediate corollary of this theorem is that we can express the
Peierls\,--\,De\,Witt bracket between two local functionals associated
to hamiltonian $(n\!-\!1)$-forms directly in terms of their ``multi%
symplectic Poisson bracket'':
\begin{thm}~
 Suppose we are given a fiber bundle\/~$E$ (the field configuration
 bundle) over an \linebreak $n$-dimensional globally hyperbolic
 space-time manifold\/~$M$ and a hamiltonian $\, \mathcal{H}:
 \vec{J}^{\circledast} E \longrightarrow J^{\circledstar} E$, which
 is a section of extended multiphase space\/~$J^{\circledstar} E$
 over ordinary multiphase space\/~$\vec{J}^{\circledast} E$,
 together with two hamiltonian\/ $(n\!-\!1)$-forms\/ $f$ and\/~$g$
 on\/~$\vec{J}^{\circledast} E$ of spatially compact support such
 that the corresponding hamiltonian vector fields\/ $X_f^{}$
 and\/~$X_g^{}$ on\/~$\vec{J}^{\circledast} E$ are projectable
 to~$M$.\footnotemark[10] \linebreak
 Then given any Cauchy surface\/~$\Sigma$ in\/~$M$, the
 Peierls\,--\,De\,Witt bracket between the local functionals
 $\mathcal{F}_{\Sigma,f}^{}$ and $\mathcal{F}_{\Sigma,g}^{}$ on\/~%
 $\mathscr{S}$  associated to\/~$\Sigma$ and\/~$f$ and to\/~$\Sigma$
 and\/~$g$, as in Definition~\ref{def:LFUNCT}, is the local functional
 associated to~$\Sigma$ and any of the multisymplectic Poisson brackets
 $\{f,g\}$ on\/~$\vec{J}^{\circledast} E$ that can be found in the
 literature, among them the simple ``pseudo-bracket'' defined by
 equation~(\ref{eq:POISB3}) as well as the modified bracket defined
 by equation~(\ref{eq:POISB4}).
 In other words, with any one of these choices, we have
 \begin{equation} \label{eq:POISB8}
  \bigl\{ \mathcal{F}_{\Sigma,f}^{} \,,\, \mathcal{F}_{\Sigma,g}^{} \bigr\}~
  =~\mathcal{F}_{\Sigma,\{f,g\}}^{}~.
 \end{equation}
\end{thm}
\begin{proof}
 Combining equation~(\ref{eq:POISB5}) with equation~(\ref{eq:FHVF3})
 from the previous theorem and applying equation~(\ref{eq:FFUNDER2})
 from Proposition~\ref{prop:FFUNDER} (which is applicable because the
 relevant integral need only be extended over a compact subset of~%
 $\Sigma$ such that the base supports of~$f$ and~$g$ are contained
 in its interior), we obtain, for any $\, \phi \in \mathscr{S}$,
 \begin{eqnarray*}
  \bigl\{ \mathcal{F}_{\Sigma,f}^{} \,,\, \mathcal{F}_{\Sigma,g}^{} \bigr\}
  [\phi] \!\!
  &=&\!\!\! \mbox{} - \, \mathcal{F}_{\Sigma,g}^{\>\prime}[\phi]
            \Bigl( \mathslf{X}_{\mathcal{F}_{\Sigma,f}}^{} [\phi] \Bigr)~
   = \; -\, \mathcal{F}_{\Sigma,g}^{\>\prime}[\phi]
            \Bigl( \delta_{X_f}^{} \phi \Bigr) \\[2mm]
  &=&\!\!\! \mbox{} - \, \int_\Sigma
            \Bigl( \phi^* \bigl( i_{X_f}^{} dg \bigr) \, - \,
                                 i_{X_{f,M}}^{} \bigl( \phi^* dg \bigr) \Bigr)~.
 \end{eqnarray*}
 Using that $\, dg = i_{X_g}^{} \omega_{\mathcal{H}}^{}$, we see that the
 second integral vanishes due to the equations of motion~(\ref{eq:FEQMOT2}),
 and we get
 \[
  \bigl\{ \mathcal{F}_{\Sigma,f}^{} \,,\,
          \mathcal{F}_{\Sigma,g}^{} \bigr\} [\phi]~
  =~\int_\Sigma \phi^*
    \Bigl( i_{X_g}^{} i_{X_f}^{} \omega_{\mathcal{H}}^{} \, \Bigr)~
  =~\int_\Sigma \phi^* \{f,g\}~
  =~\mathcal{F}_{\Sigma,\{f,g\}}^{}[\phi]~,
 \]
 where the second equality is obvious if we employ the ``pseudo-bracket''
 of equation~(\ref{eq:POISB3}) but holds equally well if we employ the
 modified bracket of equation~(\ref{eq:POISB4}) since, once again, the
 relevant integral need only be extended over a compact subset of~%
 $\Sigma$ such that the base supports of~$f$ and~$g$ are contained in
 its interior and then the integral over the additional term can, by
 Stokes's theorem, be converted to an integral over the boundary of
 that compact subset and hence vanishes.
\qed
\end{proof}

\section{Conclusions and Outlook}

In this paper we have established a link between the multisymplectic
Poisson brackets that have been studied by geometers over more than
four decades and the covariant functional Poisson bracket of classical
field theory, commonly known as the Peierls\,--\,De\,Witt bracket.
This link is based on associating to each differential form on the
pertinent multiphase space a certain local functional obtained by
pulling that form back to space-time via a solution of the equations
of motion and integrating over some fixed submanifold of space-time
of the appropriate dimension, considering the result as a functional
on covariant phase space (the space of solutions of the equations
of motion). 
Here, we have restricted attention to forms of degree $n\!-\!1$,
where $n$ is the dimension of space-time, which have to be integrated
over submanifolds of codimension~$1$ (hypersurfaces), but we cannot
see any obvious obstruction to extending this kind of analysis to
forms of other degree.
This would be of considerable interest since in physics there appear
many functionals that are localized on submanifolds of space-time of
other dimensions, such as: values of observable fields at space-time
points (dimension~$0$), Wilson loops (traces of parallel transport
operators around loops) in gauge theories (dimension~$1$),
electromagnetic field strength tensors and curvature tensors
(dimension~$2$), etc.

The overall picture that emerges is that the correct approach to
the concept of observables \linebreak in classical field theory is
to regard them as smooth functionals on covariant phase space.
\linebreak
As is well known, covariant phase space is, at least formally and for
nondegenerate systems, an (infinite-dimensional) symplectic manifold,
so the space of all such functionals constitutes a Poisson algebra,
and that is what we are referring to when we speak about the
``algebra of \linebreak observables'' in classical field theory.
Of course, this algebra is huge, and the construction of local
functionals on covariant phase space from differential forms on
multiphase space, as employed in this paper, is merely a device for
producing special (and quite small) classes of such observables.
\linebreak
But the reduction of the algebraic structure at the level
of such functionals to some corresponding algebraic structure
for the generating differential forms is highly problematic:
in fact, there is no reason to expect that there might exist
any product or bracket between differential forms on multiphase
space capable of reproducing the standard product or bracket
between the corresponding functionals on covariant phase space.
One possible obstacle is that any such prescription would most
likely be highly ambiguous since a crucial piece of information
is missing: after all, the functional does not only depend on the
differential form which (after pull-back) is being integrated but
also on the submanifold over which one integrates!
One way out would be to restrict to functionals defined by some
fixed submanifold and hope that the resulting algebraic structure
does not depend on which submanifold (within a certain given class)
is chosen: that is what we have done in this paper when reducing the
covariant Peierls\,--\,De\,Witt bracket to a multisymplectic bracket.
But the fact that this actually works is to a certain extent a miracle
which cannot be expected to happen in general, since we will very
likely be forced into admitting functionals defined by integration
over different submanifolds, including submanifolds of different
dimension.
For example, this happens as soon as we want to include differential
forms of different degrees and/or explore the existence of a relation
between the product of functionals and the exterior product of forms,
or some modified form thereof.
In particular, these arguments show why the notorious absence of 
a decent associative product in the multi\-symplectic formalism
should come as no surprise: it merely expresses the fact that
functionals of the form given by  equation~(\ref{eq:LFUNCT1})
do not form a subalgebra.

At any rate, it is a highly interesting question what kind of
algebraic structure on what kind of spaces of differential forms
(or pairs of submanifolds and differential forms) will ultimately
result from the functional approach advocated in this paper.
These and similar questions are presently under investigation
and will be reported in a future publication.

{\footnotesize

\end{document}